\documentclass[12pt]{article}

\usepackage{amsmath,amsthm,amssymb,mathtools}
\usepackage{hyperref}
\usepackage[authoryear]{natbib}

\usepackage{booktabs} 
\usepackage[ruled,linesnumbered]{algorithm2e} 
\SetCommentSty{ttfamily} 

\SetAlFnt{\small}
\SetAlCapFnt{\small}
\SetAlCapNameFnt{\small}
\SetAlCapHSkip{0pt}
\IncMargin{-\parindent}

\usepackage[letterpaper,margin=1in]{geometry}

\def\obs{\mathrm{obs}}

\def\Ber{\mathrm{Bernoulli}}
\def\Bin{\mathrm{Bin}}

\usepackage{bbm}

\newtheorem{theorem}{Theorem}[section]
\newtheorem{corollary}[theorem]{Corollary}
\newtheorem{lemma}[theorem]{Lemma}

\newtheorem{proposition}[theorem]{Proposition}

\newenvironment{fminipage}%
  {\begin{Sbox}\begin{minipage}}%
  {\end{minipage}\end{Sbox}\fbox{\TheSbox}}

\def\defeq{\stackrel{\mathrm{def}}{=}}

\def\abs#1{\left|#1  \right|}

\def\eps{\epsilon}

\renewcommand{\arraystretch}{1}

\newcommand\bb{\boldsymbol{\mathit{b}}}
\newcommand\cc{\boldsymbol{\mathit{c}}}

\newcommand\vv{\boldsymbol{\mathit{v}}}

\newcommand\xx{\boldsymbol{\mathit{x}}}

\renewcommand\AA{\boldsymbol{\mathit{A}}}

\newcommand\NN{\boldsymbol{\mathit{N}}}

\newcommand\WW{\boldsymbol{\mathit{W}}}

\newcommand\YY{\boldsymbol{\mathit{Y}}}

\newcommand\ZZ{\boldsymbol{\mathit{Z}}}

\title{Fast Algorithms for Exact Confidence Intervals in Randomized Experiments with Binary Outcomes}

\author{Peng Zhang\\
Department of Computer Science\\
Rutgers University\\
\texttt{pz149@rutgers.edu}}

\begin{document}

\maketitle

\begin{abstract}

We construct exact confidence intervals for the average treatment effect in randomized experiments with binary outcomes using sequences of randomization
tests. Our approach does not rely on large-sample approximations and is valid for all sample sizes. 
Under a balanced Bernoulli design or a matched-pairs design, we show that exact confidence intervals can be computed using only $O(\log n)$ randomization tests,
yielding an exponential reduction in the number of tests compared to brute-force. We further prove an information-theoretic lower bound showing
that this rate is optimal.
In contrast, under balanced complete randomization, the most efficient known procedures require $O(n\log n)$ randomization tests \citep{aronow2023fast},
establishing a sharp separation between these designs.
In addition, we extend our algorithm to general Bernoulli designs using $O(n^2)$ randomization tests.
A Python implementation of the algorithms described in the paper is available at \texttt{https://github.com/pengzhang91/randomization-ci}.

\end{abstract}


\begin{titlepage}


\vspace{1cm}
\setcounter{tocdepth}{2} 

\end{titlepage}


\section{Introduction}

Randomized experiments are the gold standard for causal effect estimation across a wide
range of fields, including medicine, education, economics, and computer science
\citep{IR15}. Many such experiments have binary outcomes, for example, whether a person
dies of cancer by a fixed follow-up time in a screening trial, whether a student graduates following an educational intervention, or whether a user completes a purchase in an online A/B test \citep{beckie2009selecting,leon2010transcatheter,sahni2015effect}.

A key challenge in the statistical analysis of randomized experiments is constructing \emph{confidence intervals} for unknown causal parameters, such as the average treatment effect in a finite population.

Classical methods for confidence interval construction often rely on asymptotic
approximations, such as the central limit theorem.
These methods can be inaccurate when sample sizes are small, outcome distributions are highly skewed, or treatment effects are heterogeneous \citep{brown2001interval}. 
In contrast, finite-sample exact confidence intervals based on randomization inference depend only on the known assignment mechanism and avoid additional modeling or asymptotic assumptions.
Such intervals are valid for any sample size, making them especially appealing when coverage guarantees are important. 

The main obstacle to exact confidence intervals is \emph{computational efficiency}: such
intervals are typically constructed by inverting randomization tests over many compatible potential outcome configurations. Because the search
space grows exponentially, naive implementations can be prohibitively expensive.

In this paper, we address this algorithmic bottleneck in the setting of binary outcomes.
We develop fast algorithms for constructing exact confidence intervals under the Bernoulli
design and the matched-pairs design. We show that, under either a balanced Bernoulli design
or a matched-pairs design, exact confidence intervals can be constructed using only
$O(\log n)$ sharp-null randomization tests\footnote{We use ``randomization test'' for tests based on the assignment mechanism; under complete randomization or matched pairs, this coincides with a permutation test, while under Bernoulli design, it corresponds to re-randomization with Bernoulli weights.}, yielding an exponential speedup relative to brute-force and making exact inference practical for larger-scale randomized experiments.
We further establish an information-theoretic lower bound, showing
that this rate is optimal. In contrast, under a balanced completely randomized design, the most
efficient known procedures require $O(n \log n)$ randomization tests
\citep{aronow2023fast}, yielding a sharp separation between these designs.
Our result on the matched-pairs design also answers an open question in \citet{aronow2023fast}.

\subsection{Problem Setup and Background}

We consider a randomized experiment with $n$ units and two arms/groups: treatment and control. 
We work within the Neyman-Rubin \emph{potential outcomes} framework \citep{rubin2005causal}. 
For each unit $i$, there are two binary potential outcomes, $Y_i(1)$ under treatment and 
$Y_i(0)$ under control, each taking values in $\{0,1\}$. 
Units are randomly assigned to either treatment or control, under a specified assignment mechanism $\mathcal{D}$.
Let $Z_i$ denote the treatment assignment for unit $i$, where $Z_i = 1$ indicates assignment 
to treatment and $Z_i = 0$ indicates assignment to control. 
Only the outcome under the realized assignment is observed, that is, $Y_i^{\obs} = Y_i(Z_i^{\obs})$. 
Let $\YY = \{(Y_i(1), Y_i(0))\}_{i \in [n]}$ and $\ZZ = (Z_1, \ldots, Z_n)$.

We treat the potential outcomes $\YY$ as fixed, and all randomness in the experiment comes from the assignment vector $\ZZ$. 
We assume the stable unit treatment value assumption (SUTVA) \citep{IR15}: each unit's potential outcomes are unaffected by the treatment assignments of other units, and for each unit and each treatment level, there are no alternative versions of the treatment that would lead to different potential outcomes.

Our goal is to estimate the \emph{sample average treatment effect} (ATE), defined by
\begin{align*}
\tau = \tau(\YY) \defeq \frac{1}{n} \sum_{i=1}^n \bigl( Y_i(1) - Y_i(0) \bigr).
\end{align*}
For a given confidence level $\alpha \in (0,1)$, we aim to construct a $(1-\alpha)$ \emph{confidence interval}, denoted $\mathcal{I}_{\alpha}(\YY, \ZZ)$, such that
\[
\Pr_{\ZZ \sim \mathcal{D}}\!\left( \tau \in \mathcal{I}_{\alpha}(\YY, \ZZ) \right)
\ge 1-\alpha .
\]

Under the assumption of binary potential outcomes, the ATE $\tau$ takes values on a finite
grid in $[-1,1]$ with spacing $1/n$.
Given the observed data $(\YY^{\obs}, \ZZ^{\obs})$, $\tau$ is further restricted to a
finite set consisting of $n+1$ possible values \citep{rigdon2015randomization}:
\begin{align}
\begin{split}
& \mathcal{T} = \mathcal{T}(\YY^\obs, \ZZ^\obs) =  \left\{ \frac{S - n_1}{n}, \ \frac{S - n_1+1}{n}, \ \ldots, \ \frac{S - n_1 + n-1}{n}, \ \frac{S - n_1 + n}{n}  \right\},  \\
& \text{where } S = \sum_{i=1}^n Y_i^{\obs} (2Z_i^{\obs} - 1), \quad
n_1 = \sum_{i=1}^n Z_i^{\obs}.
\end{split} 
\label{eqn:tau_set}
\end{align}
Since $\tau$ lies in the finite set $\mathcal{T}(\YY^{\obs}, \ZZ^{\obs})$, we report a \emph{confidence set} $\mathcal{I}_\alpha (\YY^{\obs}, \ZZ^{\obs}) \subseteq \mathcal{T}$ given the observed data. When the context is clear, we simply write $\mathcal{I}_\alpha = \mathcal{I}_\alpha (\YY^{\obs}, \ZZ^{\obs})$. 
When convenient, we summarize $\mathcal{I}_\alpha$ by its endpoints and get an interval $[\min \mathcal{I}_\alpha, \, \max \mathcal{I}_\alpha]$.

\subsubsection{Assignment Mechanism}

In this paper, we focus on two widely used treatment assignment mechanisms: the Bernoulli
design and the matched-pairs design.
Under the Bernoulli design, each unit $i$ is independently assigned to treatment with
probability $p$ and to control with probability $1-p$. When $p = 1/2$, the design is
\emph{balanced}, since the treatment and control group sizes are equal in expectation;
otherwise, it is \emph{unbalanced}. 
Under the matched-pairs design, assuming $n$ is even, units are first partitioned into pairs
based on covariate similarity\footnote{In this setting, pre-treatment covariates are observed, and the matched-pairs design pairs units with similar covariates.}. Within each pair, one unit is randomly assigned to treatment and the other to
control with equal probability, and assignments are independent across pairs
\citep{greevy2004optimal}.

Another commonly used assignment mechanism is the completely randomized design or complete randomization, in which
$m$ of the $n$ units are randomly selected to receive treatment and the remaining $n-m$
units are assigned to control. Fast algorithms for constructing exact confidence intervals
have primarily been developed for complete randomization
\citep{rigdon2015randomization,li2016exact,aronow2023fast}. Unlike the Bernoulli and
matched-pairs designs, the completely randomized design induces dependence among unit or
unit-pair assignments.

The independence structure inherent to the Bernoulli and matched-pairs designs
enables a qualitatively stronger reduction in the number of randomization tests required for
constructing exact confidence intervals.

\subsection{Exact Confidence Interval Construction by Inverting Randomization Tests}
\label{sec:randomization_test_inversion}

\citet{rigdon2015randomization} propose a randomization-based approach for constructing
a confidence set \(\mathcal{I}_{\alpha} \subset \mathcal{T}\).
Given the observed data \((\YY^{\obs}, \ZZ^{\obs})\) and a candidate value
\(\tau_0 \in \mathcal{T}\), one can complete a potential outcome table $\tilde{\YY} \in \{0,1\}^{2n}$ by imputing all the $n$ missing potential outcomes. 
We say a completion $\tilde{\YY}$ is \emph{compatible} with $(\tau_0, \YY^{\obs}, \ZZ^{\obs})$ if:
(1) \(\tilde{\YY}\) agrees with the observed outcomes, meaning that
\(\tilde{Y}_i(1) = Y_i^{\obs}\) if \(Z_i^{\obs} = 1\) and
\(\tilde{Y}_i(0) = Y_i^{\obs}\) if \(Z_i^{\obs} = 0\); and
(2) the ATE under \(\tilde{\YY}\) equals \(\tau_0\), that is,
\(\tau(\tilde{\YY}) = \tau_0\).
We denote the set of all compatible outcome completions by
\(\mathcal{Y}(\tau_0, \YY^{\obs}, \ZZ^{\obs})\).

A candidate value $\tau_0$ is included in $\mathcal{I}_{\alpha}$ if there exists a $\tilde{\YY} \in \mathcal{Y}(\tau_0, \YY^{\obs}, \ZZ^{\obs})$ accepted by a randomization test for the sharp null hypothesis $H_0(\tilde{\YY}): \YY_{\text{true}} = \tilde{\YY}$, as defined below.
A randomization test computes how extreme a chosen test statistic is under the assignment distribution
$\mathcal{D}$ when the sharp null holds, and uses it to accept or reject the hypothesis.

For each $\tilde{\YY} \in \mathcal{Y}(\tau_0, \YY^{\obs}, \ZZ^{\obs})$, we test $H_0(\tilde{\YY}): \YY_{\mathrm{true}} = \tilde{\YY}$
using the Horvitz-Thompson (HT) estimator as the test statistic\footnote{We note that \citet{rigdon2015randomization} use the difference-in-means (DiM) estimator under a completely randomized design. When the treatment and control group sizes are fixed, the DiM estimator coincides with the HT estimator. In contrast, when the number of treated units is random, the DiM estimator may be biased.}:
\[
T(\tilde{\YY}, \ZZ)
\defeq \frac{1}{n}
\left(
\sum_{i: Z_i = 1} \frac{\tilde Y_i(1)}{p}
-
\sum_{i: Z_i = 0} \frac{\tilde Y_i(0)}{1 - p}
\right).
\]
The HT estimator is unbiased, satisfying $\mathbb{E}_{\ZZ \sim \mathcal{D}}[T(\tilde{\YY}, \ZZ)] = \tau(\tilde{\YY})$.
Under $H_0(\tilde{\YY})$, all potential outcomes are known, and thus the distribution of
$T(\tilde{\YY}, \ZZ)$ under the assignment distribution $\mathcal{D}$ is completely specified.
The randomization test \emph{$p$-value} is
\[
p(\tilde{\YY}, \ZZ^{\obs})
\defeq
\Pr_{\tilde{\ZZ} \sim \mathcal{D}}
\left(
\bigl| T(\tilde{\YY}, \tilde{\ZZ}) - \tau(\tilde{\YY}) \bigr|
\ge
\bigl| T(\tilde{\YY}, \ZZ^{\obs}) - \tau(\tilde{\YY}) \bigr|
\right).
\]
It computes how extreme the observed centered test statistic is relative to its randomization distribution under the design $\mathcal{D}$.
We accept $H_0(\tilde{\YY})$ at level $\alpha$ if $p(\tilde{\YY}, \ZZ^{\obs}) \ge \alpha$, and reject otherwise.

Equivalently,
\begin{equation}
\mathcal{I}_{\alpha}
=
\bigl\{\tau_0 \in \mathcal{T} :\; p_{\max}(\tau_0) \ge \alpha \bigr\},
\qquad
p_{\max}(\tau_0)
\defeq
\max_{\tilde{\YY} \in \mathcal{Y}(\tau_0, \YY^{\obs}, \ZZ^{\obs})}
p(\tilde{\YY}, \ZZ^{\obs}).
\label{eqn:ci}
\end{equation}
Standard argument shows that $\mathcal{I}_\alpha$ is a valid $(1-\alpha)$ confidence set.

The quantity $p_{\max}(\tau_0)$ can be interpreted as the optimal value of a
\emph{constrained discrete optimization problem} over $n$ binary variables, corresponding to the $n$ missing outcomes under $(\YY^{\obs}, \ZZ^{\obs})$.

\subsubsection{Computational Challenge and Previous Works}

The main challenge in constructing \(\mathcal{I}_{\alpha}\) is its computational
cost. A brute-force procedure would enumerate all $2^n$ configurations of the unobserved outcomes to complete
the potential-outcome table $\tilde{\YY}$, and then run a randomization test for each completion.

\citet{rigdon2015randomization} show that any completion
\(\tilde{\YY}\) is fully summarized by the four counts of joint potential-outcome
types \((Y(1),Y(0))\in\{0,1\}^2\) (formal notation deferred to Section~\ref{subsec:ht_property}).
Consequently, the search over \(2^n\) completions can be replaced by a search over
\(O(n^4)\) feasible count configurations.

Subsequent work has aimed to reduce the number of randomization tests needed under the completely randomized design by leveraging monotonicity properties of the set of compatible potential outcomes.
In the balanced setting, where the treatment and control groups have equal size, \citet{li2016exact} reduce the number of randomization tests to $O(n^2)$, and \citet{aronow2023fast} further improve this to $O(n\log n)$.
In the general (possibly unbalanced) case, the fastest known constructions require $O(n^2)$ randomization tests
\citep{aronow2023fast}.

The above exact confidence-interval methods extend to stratified randomization, where units are partitioned into disjoint strata and, within each stratum, independently randomized to treatment or control \citep{rigdon2015randomization,chiba2017stratified,li2025exact}.
However, the computational cost is high: \citet{li2025exact} show that the number of randomization tests required is $O(\sum_{k=1}^K n_k \times \prod_{k=1}^K n_k^2)$ when all strata are balanced, and $O(\prod_{k=1}^K n_k^3)$ otherwise, where $n_k$ is the size of the $k$th stratum.
Viewing a matched-pairs design as the special case of stratified randomization in which every stratum has size two, these results imply that $O(n \, 2^n)$ randomization tests are required.

Exact confidence intervals obtained by inverting randomization tests have also been studied under the Bernoulli design. However, existing work neither characterizes the exact worst-case $p$-value over feasible binary potential-outcome configurations nor provides computational complexity bounds for constructing such exact intervals \citep{basu1980randomization,branson2019randomization}.

\subsection{Our Results}

We develop fast algorithms for constructing exact confidence intervals via inversion of randomization tests under both the Bernoulli design and the matched-pairs design.
Bernoulli designs are commonly used when
treatment assignment must be implemented independently across units, such as decentralized
experiments or online platforms with asynchronous enrollment. Moreover, they minimize the worst-case mean-squared error of the HT estimator for the ATE
\citep{HSSZ24}.
Matched-pairs designs are common when pre-treatment covariates are available and pairing similar units can substantially improve precision \citep{IR15,bai2022optimality}.

\begin{theorem}
For any $\alpha \in (0,1)$ and observed data $(\YY^{\obs}, \ZZ^{\obs})$,
under the balanced Bernoulli design or the matched-pairs design, the confidence set $\mathcal{I}_\alpha(\YY^{\obs}, \ZZ^{\obs})$ can be constructed using at most $8 \log_2 n$ randomization tests. 
\label{thm:upper}
\end{theorem}

Our theorem provides an exponential reduction in the number of randomization tests compared with brute-force enumeration over the $\Theta(n^4)$ potential-outcome count vectors.
It also highlights a separation between balanced Bernoulli, matched-pairs designs and balanced complete randomization, for which the fastest known procedure uses $O(n \log n)$ randomization tests \citep{aronow2023fast}.
Moreover, we answer an open question raised in \citet{aronow2023fast} on how to generalize the fast randomization-based constructions to matched-pairs designs.

The stronger reduction in the number of randomization tests for balanced Bernoulli and matched-pairs designs comes from their \emph{symmetry and independence} across units or unit pairs.
These structures enable us to express the centered test statistic as a sum of independent Rademacher and half-Rademacher random variables, a special case of Poisson binomial distribution.
Maximizing the $p$-value over all compatible potential outcomes $\mathcal{Y}(\tau_0,\YY^{\obs},\ZZ^{\obs})$ for a fixed $\tau_0$ amounts to choosing a feasible outcome-count vector that makes the test statistic look as non-extreme as possible under the induced randomization distribution. 
This corresponds to pushing that distribution toward a more ``dispersed'' law.
We show that under balanced Bernoulli and matched pairs, this most dispersed law (and thus the maximum $p$-value) is attained at one of two boundary points of the set of compatible outcome-count vectors, both of which can be identified in constant time.
In contrast, under the balanced complete randomization, one needs to search over $O(n)$ boundary points to find the maximum $p$-value. 
Finally, we show that we can binary search the ATE candidate set $\mathcal{T}(\YY^\obs, \ZZ^\obs)$ to construct the confidence set $\mathcal{I}_\alpha$, which is in analogy with the balanced complete randomization \citep{aronow2023fast}.

We also establish an information-theoretic lower bound: in the worst case, computing the exact confidence set $\mathcal{I}_\alpha$ requires at least $\Omega(\log n)$ randomization tests.

\begin{theorem}
For any design, any deterministic algorithm that takes as input observed data $(\YY^{\obs}, \ZZ^{\obs})$ and a confidence level $\alpha \in (0,1)$, and constructs the $(1-\alpha)$ confidence set $\mathcal{I}_\alpha=\{\tau_0 \in \mathcal{T}(\YY^{\obs}, \ZZ^{\obs}): p_{\max}(\tau_0) \ge \alpha\}$ by inverting randomization tests,
must perform $\Omega(\log n)$ randomization tests in the worst case.
\label{thm:lower}
\end{theorem}

The total computational cost equals the product of the number of randomization tests inverted and the cost per randomization test (equivalently, per $p$-value evaluation).
Under the Bernoulli and matched-pairs designs, we show that computing the $p$-value for a given potential-outcome configuration or count vector reduces to evaluating tail probabilities (equivalently, the CDF) of a Poisson binomial distribution. These probabilities can be computed exactly via the Fast Fourier Transform \citep{hong2013computing} in $O(n \log n)$ time \citep{cooley1965algorithm,thomas2009introduction}, avoiding Monte Carlo simulations. 
This is analogous to \citet{li2016exact}, who evaluate $p$-values under complete randomization using hypergeometric tail probabilities.
Throughout, we analyze all algorithms in the standard unit-cost RAM model, in which each arithmetic operation takes constant time \citep{thomas2009introduction}.

\begin{theorem}
For any $\alpha \in (0,1)$ and observed data $(\YY^{\obs}, \ZZ^{\obs})$, under the balanced Bernoulli design or the matched-pairs design, the confidence set $\mathcal{I}_\alpha(\YY^{\obs}, \ZZ^{\obs})$ can be constructed in time $O(n \log^2 n)$. 
\label{thm:time}
\end{theorem}

\citet{aronow2023fast} consider approximating randomization-test $p$-values via Monte Carlo, which for balanced complete randomization yields a running time of $O(Kn\log n)$, where $K$ is the number of simulation draws per test; their experiments suggest taking $K \ge 10^4$ for moderate $n$.
In contrast, our FFT-based computation is deterministic and can be implemented efficiently in practice; specifically, it needs at most $2 N\log_2 N$ arithmetic operations for $N$ points
\citep{cooley1965algorithm,singleton1969algorithm,hong2013computing,weisstein2015fast}, which is smaller than $10^4$ for moderate sample sizes.

We further extend our results on the balanced Bernoulli design to general settings where the treatment assignment probability $p \in (0,1)$.
Under this general case, we establish a result matching the number of randomization tests needed for the general completely randomized design in \citep{aronow2023fast}. We adapt their idea of reusing randomization test results (more specifically, $p$-values) for ``nearly'' outcome configurations.

\begin{theorem}
For any $\alpha \in (0,1)$ and observed data $(\YY^\obs, \ZZ^\obs)$, under the Bernoulli design, the confidence set $\mathcal{I}_\alpha(\YY^\obs, \ZZ^\obs)$ can be constructed using $O(n^2)$ randomization tests and in time $O(n^3 \log n)$ via FFT. 
\label{thm:general}
\end{theorem}

\paragraph{Roadmap.}
Section~\ref{sec:binary_search} establishes basic properties of the balanced Bernoulli design.
Section~\ref{sec:four_permutation} presents our main algorithm, which computes the maximum \(p\)-value under a given ATE $\tau_0$ using at most two randomization tests.
Section~\ref{sec:match} extends these results to the matched-pairs design.
Section~\ref{sec:simulation} reports simulations comparing our algorithms with common benchmarks.
We discuss the general Bernoulli design and prove Theorem \ref{thm:general} in Appendix \ref{sec:general_bernoulli}.
All missing proofs are in Appendix \ref{sec:appendix_proofs}.

\section{Binary Search, Lower Bounds, and Convergence Rates}
\label{sec:binary_search}

In this and the next sections, we focus on the balanced Bernoulli design, where $Z_i$'s are independent Bernoulli random variables with $\Pr(Z_i = 1) = \Pr(Z_i = 0) = 1/2$ for all $i$.

We first show that, given the observed data \((\YY^{\obs}, \ZZ^{\obs})\) and level \(\alpha\), one can
perform binary search over the candidate set \(\mathcal T  = \mathcal{T}(\YY^{\obs}, \ZZ^{\obs})\) to identify the
confidence set \(\mathcal{I}_\alpha = \mathcal{I}_\alpha (\YY^{\obs}, \ZZ^{\obs})\).
Let \(T^{\obs}\) denote the observed test statistic.
For every potential-outcome completion \(\tilde{\YY}\) compatible with the data,
the value of \(T(\tilde{\YY}, \ZZ^{\obs})\) is identical and equals \(T^{\obs}\).
Observe that \(p_{\max}(T^{\obs})=1\), which implies
\(\min \mathcal{I}_\alpha \le T^{\obs} \le \max \mathcal{I}_\alpha\).
The next lemma characterizes the monotonicity of \(p_{\max}\) on either side of
\(T^{\obs}\).

\begin{lemma}
Under the balanced Bernoulli design, the function $p_{\max}$ is monotone nondecreasing on
$\mathcal{T} \cap [-1, T^{\obs}]$ and monotone nonincreasing on
$\mathcal{T} \cap [T^{\obs}, 1]$.
\label{lem:f_monotone}
\end{lemma}

We remark that $T^{\obs}$ need not lie in $\mathcal{T}$.
For example, if all units are assigned to treatment and all observed outcomes equal $1$,
then $T^{\obs} = 2$, which falls outside the set $\mathcal{T}$.

\begin{corollary}
Under the balanced Bernoulli design, the
confidence set $\mathcal{I}_{\alpha}$ can be constructed by binary search over each of
$\mathcal{T} \cap [-1, T^{\obs}]$ and $\mathcal{T} \cap [T^{\obs}, 1]$ as in Algorithm \ref{alg:bernoulli_balance} (details in Algorithm \ref{alg:binary_search} in Appendix \ref{sec:appendix_proofs}), requiring evaluation
of $p_{\max}$ at no more than $4 \log_2 n$ distinct values.
\label{coro:binary_search}
\end{corollary}

Corollary~\ref{coro:binary_search} gives an upper bound on the number of evaluations
of \(p_{\max}\) needed to construct \(\mathcal{I}_{\alpha}\).
We next show that this bound is tight up to constant factors by proving a matching
information-theoretic lower bound.

We study a simplified oracle model: an algorithm may query any
\(\tau_0 \in \mathcal{T}\), and the oracle returns a single indicator bit
\(\mathbbm{1}\{p_{\max}(\tau_0)\ge \alpha\}\).
This oracle can be implemented via sharp-null randomization tests, where each query
requires at least one such test.

The lemma below lower-bounds the number of oracle queries required by any
deterministic algorithm; consequently, it also lower-bounds the number of
randomization tests needed to invert and construct \(\mathcal{I}_{\alpha}\).
That is, this lemma implies Theorem~\ref{thm:lower}.

\begin{algorithm}[t]

\SetAlgoNoLine
\DontPrintSemicolon

\KwIn{Observed data $(\YY^{\obs}, \ZZ^{\obs})$ and confidence level $\alpha \in (0,1)$.}
\KwOut{A $(1-\alpha)$ confidence set $\mathcal{I}_\alpha$.}

We perform a binary search over $\tau_0 \in \mathcal{T}(\YY^{\obs}, \ZZ^{\obs})$ described in detail in Algorithm \ref{alg:binary_search}, using a function $f: \mathcal{T}(\YY^{\obs}, \ZZ^{\obs}) \rightarrow \{0,1\}$
defined by 
\[
f(\tau_0) = \begin{cases}
1, \quad & \text{if } p_{\max}(\tau_0) \ge \alpha    \\
0, \quad & \text{if } p_{\max}(\tau_0) < \alpha 
\end{cases}
\]
where 
$p_{\max}(\tau_0) \gets \textsc{ComputePmax}((\YY^{\obs}, \ZZ^{\obs}), \tau_0)$ is computed using Algorithm \ref{alg:p_max}.

	\caption{Confidence Set for Balanced Bernoulli}
	\label{alg:bernoulli_balance}
\end{algorithm}

\begin{lemma}
For any design, any deterministic algorithm that takes as input observed data $(\YY^{\obs}, \ZZ^{\obs})$ and a confidence level $\alpha \in (0,1)$, and constructs the $(1-\alpha)$ confidence set $\mathcal{I}_\alpha=\{\tau_0 \in \mathcal{T}(\YY^{\obs}, \ZZ^{\obs}): p_{\max}(\tau_0) \ge \alpha\}$ by inverting randomization tests,
must evaluate $p_{\max}$ at 
\(\Omega(\log n)\) distinct values in the worst case.
\label{lem:lower}
\end{lemma}

The following proposition shows that, under the balanced Bernoulli design, the confidence interval
\([\min \mathcal{I}_\alpha,\, \max \mathcal{I}_\alpha]\) converges at the rate
\(n^{-1/2}\). This convergence rate extends to general Bernoulli designs
with marginal treatment assignment probability $p \in (0,1)$ bounded away from $0$ and $1$.
Such a rate is characteristic of confidence intervals derived from the central limit theorem asymptotics. 
Our result parallels the \(n^{-1/2}\) convergence rate
established by \citet{aronow2023fast} for the completely randomized design.

\begin{proposition}
\label{prop:bern-hoeffding-length}

For any given observed data $(\YY^{\obs},\ZZ^{\obs})$, under the balanced Bernoulli design, and confidence level $\alpha \in (0,1)$, the confidence set $\mathcal{I}_{\alpha}$ returned by Algorithm \ref{alg:bernoulli_balance} satisfies
\[
\max \mathcal{I}_{\alpha} - \min \mathcal{I}_{\alpha} \le 
\sqrt{\frac{32\log(2/\alpha)}{n}},
\]
where $\log$ has base $e$. 
\end{proposition}

\section{Computing $p_{\max}(\tau_0)$ Using at Most Two Randomization Tests for Balanced Bernoulli}
\label{sec:four_permutation}

This section shows that, under the balanced Bernoulli design, for any given \(\tau_0 \in \mathcal{T}\), we can evaluate the exact
value of \(p_{\max}(\tau_0)\) using \emph{at most two} randomization tests.
Combining with Corollary \ref{coro:binary_search}, we prove Theorem \ref{thm:upper} for the balanced Bernoulli design.

\subsection{Monotonicity Properties of the Centered HT Estimator}
\label{subsec:ht_property}

We can represent a potential-outcome table by the four joint potential-outcome counts
\begin{align*}
v_{yy'} \defeq \abs{\{\, i \in [n] : Y_i(1)=y,\; Y_i(0)=y' \,\}},
\qquad y,y' \in \{0,1\}.
\label{eqn:potential_table}
\end{align*}
Likewise, we can summarize the observed data \((\YY^{\obs},\ZZ^{\obs})\) by the four
observation counts
\[
N_{yz}\defeq | \{\, i \in [n] : Y_i^{\obs}=y,\; Z_i^{\obs}=z \,\} |,
\qquad y,z \in \{0,1\}.
\]
Let \(\vv=(v_{11},v_{10},v_{01},v_{00})\) and \(\NN=(N_{11},N_{01},N_{10},N_{00})\),
both in \(\mathbb{Z}_{\ge 0}^{\{0,1\}^2}\).
This count-table representation is standard for
binary outcomes \citep{rigdon2015randomization, li2016exact,aronow2023fast,ding2016potential}.
We can write the feasible ATE set $\mathcal{T}(\YY^{\obs}, \ZZ^{\obs})$ in Equation \eqref{eqn:tau_set} as $\mathcal{T}(\NN)$.

A potential-outcome count \(\vv\) is \emph{feasible} with respect to a given \(\tau_0\) and \(\NN\)
if there exists some completion \(\tilde{\YY}\in\mathcal{Y}(\tau_0,\YY^{\obs},\ZZ^{\obs})\)
whose associated count is \(\vv\).
For all such \(\tilde{\YY}\), the randomization \(p\)-values \(p(\tilde{\YY},\ZZ^{\obs})\) are all the same. Hence, we may unambiguously define
\[
p(\vv, \NN) \defeq p(\tilde{\YY},\ZZ^{\obs}).
\]
In addition, for a fixed $\tau_0$ and observed data $\NN$, the threshold $|T^{\obs} - \tau_0|$ appearing in $p(\vv, \NN)$ is fixed.

The following lemma shows that  \(p(\vv, \NN)\) can be written as a tail probability for a sum of
independent Rademacher and half-Rademacher random variables.
A Rademacher random variable is uniformly distributed on \(\{\pm 1\}\), and a half-Rademacher random variable is uniformly distributed on \(\{\pm 1/2\}\).

\begin{lemma}
\label{lem:objective_in_v}
Given the observed data $\NN$ and a consistent potential-outcome count vector $\vv$, define $a \defeq v_{11}$ and $b \defeq v_{10} + v_{01}$, and let
\[
S_{a,b} = \sum_{r=1}^{a} \xi_r + \sum_{t=1}^{b} \zeta_t,
\]
where $\xi_r \sim \mathrm{Unif} \{\pm 1\}$ and $\zeta_t \sim \mathrm{Unif} \{\pm \frac{1}{2} \}$ are independent random variables.
Then, 
\[
p(\vv, \NN) = \Pr\left(|S_{a,b} | \ge \Delta \right), \qquad
\Delta \defeq \frac{n}{2}\, |T^{\obs} - \tau_0|.
\]

\end{lemma}

\begin{proof}
Let $\tilde{\YY}$ be a potential-outcome completion consistent with $\vv$ and $\NN$, and thus $\tau(\tilde{\YY}) = \tau_0$.
Let \(\tilde{\ZZ} \in \{0,1\}^n\) be an assignment vector drawn from the balanced Bernoulli design.
Equivalently, let 
\[
x_1 \sim \text{Bin}(v_{11}, 1/2), \ 
x_2 \sim \text{Bin}(v_{10}, 1/2), \ 
x_3 \sim \text{Bin}(v_{01}, 1/2)
\]
be independent binomial random variables; that is, $x_1$ units with potential outcome $(1,1)$, $x_2$ units with potential outcome $(1,0)$, and $x_3$ units with potential outcome $(0,1)$ are assigned to treatment.

Under this representation, the test statistic has the same distribution as
\begin{align*}
T(\tilde{\YY}, \tilde{\ZZ}) \overset{d}{=} \frac{2}{n} \left( x_1 + x_2 - (v_{11} - x_1) - (v_{01} - x_3) \right)   
= \frac{2}{n} \left( 2x_1 + x_2 + x_3 - v_{11} - v_{01} \right).
\end{align*}
Thus,
\begin{align*}
\frac{n}{2} \, (T(\tilde{\YY}, \tilde{\ZZ}) - \tau_0)
\overset{d}{=} 2x_1 - v_{11} + x_2 + x_3 - \frac{1}{2} (v_{10} + v_{01}),
\end{align*}
which has the same distribution as $S_{a,b}$.
\end{proof}

By Lemma \ref{lem:objective_in_v}, $p(\vv, \NN)$ depend on $\vv$ only through $(a,b)$.
We therefore reparameterize and write $p(\vv,\NN)=p(a,b;\Delta)$ and with $a=v_{11}$ and $b=v_{10}+v_{01}$.
For a fixed $\tau_0$, the pair $(a,b)$ uniquely determines $\vv$ (see Equation \eqref{eqn:v_express_in_ab}).
Since $v_{10} - v_{01} = n\tau_0$, the parameter $b$ has the same parity as $n \tau_0$.

\subsubsection{Unimodality of $S$}

A key observation is that the distribution of $S_{a,b}$ is \emph{unimodal}. 
This property allows us to establish that $p(a,b;\Delta)$ is \emph{monotone} in $a$ and $b$ along feasible directions.

Let $\Lambda_{a,b}$ denote the support of $S_{a,b}$. This set is contained either
in $\mathbb{Z}$ or in $\mathbb{Z}+1/2$, depending on the parity of $b$.
Define the tail probability function and the probability mass function:
\[
Q_{a,b}(t)\defeq \Pr(S_{a,b} \ge t), \qquad
\pi_{a,b,t}\defeq \Pr(S_{a,b}=t), \qquad t\in\Lambda_{a,b}.
\]
Note that $\Delta \in\Lambda_{a,b}$ since $(a,b)$ and the uniquely determined outcome count $\vv$ are consistent with the observed data.

\begin{lemma}
For any fixed $(a,b) \in \mathbb{Z}_{\ge 0}^2$, $S_{a,b}$ is unimodal. Specifically, for notational
simplicity, we locally drop the subscript $(a,b)$ and write $\Lambda = \Lambda_{a,b}$ and $\pi_t = \pi_{a,b,t}$.
For $t\in\Lambda$, let
\[
t^- \defeq \max\{s\in\Lambda:s<t\},\qquad
t^+ \defeq \min\{s\in\Lambda:s>t\},
\]
whenever these neighbors exist. Then there exists $M\in\Lambda$ such that
\[
\pi_t \ge \pi_{t^-}\quad \text{for all } t\in\Lambda \text{ with } t\le M,
\qquad
\pi_{t^+}\le \pi_t\quad \text{for all } t\in\Lambda \text{ with } t\ge M.
\]
In particular, if $\Lambda \subset \mathbb{Z}$, take $M = 0$; if $\Lambda \subset \mathbb{Z} + 1/2$, take $M = -1/2$.
\label{lem:unimodal}
\end{lemma}

An example is given in Table \ref{tab:unimodal}.

\begin{table}[!ht]
\centering
\caption{Probability mass functions of \(S_{a,b}\) for \((a,b)=(2,2)\) and \((a,b)=(2,1)\).}
\begin{minipage}{0.48\textwidth}
\centering
{\renewcommand{\arraystretch}{1.2}
\[
\begin{array}{|c|ccccccc|}
\hline
\Lambda_{2,2} & -3 & -2 & -1 & 0 & 1 & 2 & 3\\ \hline
\pi_{2,2} & \tfrac{1}{16} & \tfrac{1}{8} & \tfrac{3}{16} & \tfrac{1}{4} & \tfrac{3}{16} & \tfrac{1}{8} & \tfrac{1}{16} \\
\hline
\end{array}
\]
}
\end{minipage}\hfill
\begin{minipage}{0.48\textwidth}
\centering
{\renewcommand{\arraystretch}{1.2}
\[
\begin{array}{|c|cccccc|}
\hline
\Lambda_{2,1} & -\tfrac{5}{2} & -\tfrac{3}{2} & -\tfrac{1}{2} & \tfrac{1}{2} & \tfrac{3}{2} & \tfrac{5}{2}\\ \hline
\pi_{2,1} & \tfrac{1}{8} & \tfrac{1}{8} & \tfrac{1}{4} & \tfrac{1}{4} & \tfrac{1}{8} & \tfrac{1}{8} \\
\hline
\end{array}
\]
}
\end{minipage}

\label{tab:unimodal}
\end{table}

Because the structure of $\Lambda_{a,b}$, in particular, whether it is unit-spaced, 
depends on whether $b=0$, we consider the cases $b=0$ and $b\ge1$ separately.
Note that $b=0$ can occur only when $\tau_0=0$.

Note that $p(a,b; 0) = \Pr(|S_{a,b}| \ge 0) = 1 \ge \alpha$ for all $(a,b)$. In addition, since $S_{a,b}$ is symmetric about $0$, we have $p(a,b; \Delta) = 2 \Pr(S_{a,b} \ge \Delta)$ for all $\Delta > 0$.
Thus, we may assume without loss of generality that $\Delta>0$ for the remainder of the proof.

Let $\mathcal{A}$ denote the feasible set of pairs $(a,b)$.
For any $\tau_0 \in \mathcal{T}(\NN)$, the set $\mathcal{A}$ is nonempty.

\subsubsection{Monotonicity of $p(a,b; \Delta)$ for $b \ge 1$.}

By the unimodality of $S_{a,b}$, we can show that $p(a,b; \Delta)$ is monotone in $a$ and $b \ge 1$ along feasible directions.

\begin{lemma}
\label{lem:ab_monotone}
Assume $\Delta>0$ and $\tau_0 \in \mathcal{T}(\NN)$, and fix the observed table $\NN$. Let $(a,b)$ and $(a',b')$ be
feasible pairs with $b\ge1$, satisfying either\footnote{Note that a feasible $b$ must have the same parity as $n\tau_0$, and thus $b$ can change only in increments of $2$.}
\[
(a',b')=(a+1,b)
\quad\text{or}\quad
(a',b')=(a,b+2).
\]
Then
\[
p(a',b';\Delta)\ge p(a,b;\Delta).
\]
\end{lemma}

Intuitively, maximizing \(p(a,b;\Delta)\) is equivalent to choosing a feasible pair \((a,b)\)
that makes the observed deviation $\Delta$
appear as non-extreme as possible under the induced randomization distribution.
Under the balanced Bernoulli design, by Lemma~\ref{lem:objective_in_v}, this corresponds to choosing $(a,b)$ such that the distribution of $S_{a,b}$
is more dispersed. 
Lemma~\ref{lem:ab_monotone} shows that adding more such summands makes the resulting sum more dispersed.

\begin{proof}
We prove the first case \((a',b')=(a+1,b)\) and defer the second case to Appendix \ref{sec:appendix_proofs}, which follows by a similar argument.
In this setting, the random variable \(S_{a',b'}\) has the same distribution as
\(S_{a,b}+\xi\), where \(\xi\) is a Rademacher random variable and is independent of $S_{a,b}$.
Note that \(\Delta\in\Lambda_{a',b'}\).

If $\Delta + 1 \in \Lambda_{a,b}$, then 
\begin{align*}
p(a',b'; \Delta) & = 2 Q_{a',b'} (\Delta) \\
& =  \Pr(S_{a,b} + \xi \ge \Delta \mid \xi = 1) + \Pr(S_{a,b} + \xi \ge \Delta \mid \xi = -1)  \\
& = Q_{a,b}(\Delta - 1) + Q_{a,b}(\Delta + 1) \\
& = 2 Q_{a,b} (\Delta) + \pi_{a,b, \Delta-1} - \pi_{a,b, \Delta} \\
& \ge 2 Q_{a,b}(\Delta) \tag{by Lemma \ref{lem:unimodal}} \\
& = p(a,b).
\end{align*}
If $\Delta + 1 \neq \Lambda_{a,b}$, that is, $\Delta = \max \Lambda_{a,b}$, then
\[
p(a',b'; \Delta) = Q_{a,b}(\Delta - 1) = \pi_{a,b, \Delta - 1} + \pi_{a,b, \Delta} \ge 2 \pi_{a,b, \Delta} = p(a,b; \Delta).
\]
This completes the proof for the first case.
\end{proof}

Lemma~\ref{lem:ab_monotone} indicates that, subject to feasibility, the parameters $a$ and $b$ should be chosen as large as possible. 
However, the feasibility constraints on the outcome count and thus on $(a,b)$ (see Equation \eqref{eqn:constraints_in_ab}) may induce a tradeoff: increasing one parameter
can force a decrease in the other, especially along the boundary of the feasible region. 

To resolve this tradeoff, for each $x = a$ such that there exists $b \ge 1$ satisfying $(a,b) \in \mathcal{A}$, we define
\[
g(x)\defeq \max_b p(x,b;\Delta)
= p\bigl(x,b^{\max}(x);\Delta\bigr),
\]
where $b^{\max}(x)$ denotes the largest feasible value of $b$ subject to $b \ge 1$ given $a=x$. 
We then analyze how $g(x)$ behaves as $x$ varies.

\begin{lemma}
Assume $\Delta>0$ and $\tau_0 \in \mathcal{T}$, and fix the observed data $\NN$. Then $g(x)$ is nondecreasing in $x$ and hence attains its maximum at $x =  a^{\max}$, the largest $a$ such that there exists some $b \ge 1$ satisfying $(a,b) \in \mathcal{A}$.
\label{lem:g_amax}
\end{lemma}

Lemma \ref{lem:g_amax} implies that the maximum $p(a,b; \Delta)$ over $\mathcal{A} \cap \{(a,b): b \ge 1\}$ is attained at the lexicographically largest element of that set.
We include a proof sketch for this lemma and defer the full proof to Appendix \ref{sec:appendix_proofs}.

\begin{proof}[Proof sketch.]
Consider an arbitrary feasible $a$, that is, there exists $b \ge 1$ such that $(a,b) \in \mathcal{A}$, and suppose $a+1$ is also feasible.
There are two possibilities of $b^{\max}(a+1)$: 
\[
b^{\max}(a+1) = b^{\max}(a) \text{ or } b^{\max}(a+1) = b^{\max}(a) - 2 \ge 1.
\]
In the first case, by Lemma \ref{lem:ab_monotone}, $p(a,b; \Delta)$ is nondecreasing in $a$, and thus $g(a+1) \ge g(a)$.
In the second case, let $b = b^{\max}(a)$ and $b' = b^{\max}(a+1)$. We can write 
\[
S_{a,b} = R + H, \quad \text{where } H = \zeta_1 + \zeta_2.
\]
Recall that $\zeta_1$ and $\zeta_2$ are two half-Rademacher random variables and $R$ is a sum of independent Rademacher and half-Rademacher random variables.
Let $\xi$ be another Rademacher random variable that is independent of $S_{a,b}$.
Then, $S_{a',b'}$ has the same distribution as $R + \xi$.

Similar to the proof of Lemma \ref{lem:ab_monotone},
we can express both tail probabilities $\Pr(S_{a,b} \ge \Delta)$ and $\Pr(S_{a',b'} \ge \Delta)$ using tail probabilities of $R$, and then apply the unimodality of $R$ (Lemma \ref{lem:unimodal}) to prove that $\Pr(S_{a',b'} \ge \Delta) \ge \Pr(S_{a,b} \ge \Delta)$, and thus $g(a+1) \ge g(a)$.
\end{proof}

\subsubsection{Monotonicity of $p(a,b; \Delta)$ for $b = 0$}

When $b = 0$, the centered test statistic $S_{a,b}$ reduces to a sum of independent Rademacher random variables, and the following lemma establishes an analogous monotonicity property for $p(a,0; \Delta)$.

\begin{lemma}
Assume $\Delta > 0$ and fix the observed data $\NN$. Then $p(a,0; \Delta)$ is nondecreasing in $a$ and hence attains its maximum at $a = a^{\max}$, the largest $a$ such that $(a,0) \in \mathcal{A}$ (assuming exist).
\label{lem:monotone_pab_b0}
\end{lemma}

Combining Lemmas~\ref{lem:g_amax} and~\ref{lem:monotone_pab_b0}, we obtain the following characterization of the parameters $(a,b)$ that attain the maximum of $p(a,b;\Delta)$, and hence $p_{\max}(\tau_0)$.

\begin{corollary}
Assume $\Delta > 0$ and $\tau_0 \in \mathcal{T}$, and fix the observed data $\NN$. 
Define
\[
\mathcal{A}_0 \defeq \{(a,b) \in \mathcal{A} : b = 0\},
\qquad
\mathcal{A}_1 \defeq \{(a,b) \in \mathcal{A} : b \ge 1\}.
\]
Then, at least one of these sets is nonempty.
The maximum value $p_{\max}(\tau_0)$ is given by
\[
p_{\max}(\tau_0)
=
\max\!\left\{
\max\{ p(a,b;\Delta) : (a,b) \in \mathcal{A}_0 \},
\;
\max\{ p(a,b;\Delta) : (a,b) \in \mathcal{A}_1 \}
\right\},
\]
with the convention that if one of the sets is empty, the corresponding maximum is omitted.
Moreover, in each of the sets $\mathcal{A}_0$ and $\mathcal{A}_1$, the corresponding maximum is attained at the lexicographically largest element of the set.
\label{coro:pmax_tau_0}
\end{corollary}

Corollary \ref{coro:pmax_tau_0} implies that at most two randomization tests are needed to compute $p_{\max}(\tau_0)$.

We note that the set $\mathcal{A}_0$ is nonempty only if $\tau_0 = 0$. 
This is because any feasible $b$ must satisfy $b \ge \abs{n \tau_0}$, since $b = v_{10} + v_{01}$ and $v_{10} - v_{01} = n\tau_0$.
Consequently, we need only one randomization test when $\tau_0 \neq 0$.

\citet{aronow2023fast} also establish a monotonicity property for the balanced completely randomized design and leverage it to accelerate the construction of exact confidence sets. However, their approach still requires checking $\Theta(n)$ distinct outcome-count vectors to determine the maximum $p$-value for a given $\tau_0$.

\subsection{Computing the Maximum $p$-Value $p_{\max}(\tau_0)$}
\label{sec:ab_max}

We present our main algorithm in Algorithm \ref{alg:p_max}. Given the observed data $\NN$ and an ATE candidate $\tau_0$, it finds the lexicographically largest elements of $\mathcal{A}_0$ and $\mathcal{A}_1$ defined in Corollary \ref{coro:pmax_tau_0} in $O(1)$ time whenever these sets are nonempty and then returns the maximum $p$-value.

\begin{algorithm}[t]
\DontPrintSemicolon
\SetAlgoNoLine
\SetKwFunction{Fn}{Function}

\SetKwFunction{ComputeAB}{ComputeAB}
\KwIn{Observation count $\NN \in \mathbb{Z}_{\ge 0}^{\{0,1\}^2}$ and ATE candidate $\tau_0 \in \mathcal{T}(\NN)$.}
\KwOut{$p_{\max}(\tau_0)$.}

Let \( n \gets N_{11} + N_{10} + N_{01} + N_{00} \) and \(\Delta \gets \tfrac{n}{2}\lvert T^{\obs}(\NN) - \tau_0\rvert\).
Let $L,U,U_b,N_{1\star}$ be as in Equation~\eqref{eqn:additional_parameters_for_ab}.

\tcp{a helper local function that computes candidate $(a,b)$}
\Fn{\ComputeAB{$L_b$}}{

Let  
  \[
    a \gets \min \{\tfrac12(U - L_b),\ n - L_b,\ N_{1\star},\ U - N_{1\star} \}, \quad
    b \gets \max \{b \le \min\{U - 2a,\ n - a,\ U_b\}:\ b \equiv n\tau_0 \pmod 2\}.
  \]

\textbf{if} $a+b<N_{1\star}$ \textbf{then} $a\gets a-1$ and
$    b\gets b+2$.
  
  \Return $(a,\,b)$.
}

\tcp{the main function}
\If{$\tau_0 \neq 0$}{
  $(a,b)\gets \ComputeAB( \abs{n\tau_0} )$.
  
  \Return $\textsc{ComputeProbabilityViaFFT}(a,b,\Delta)$.
}

\If{$\tau_0 = 0$}{

Let $B = \emptyset$.
\tcp*{a set of candidates $(a,b)$}

Let $a \gets \min \left\{\frac{1}{2}U, \ n, \ N_{1\star}, \ U - N_{1\star} \right\}$.
\tcp*{compute $(a,b)$ under $b = 0$ if feasible}

\textbf{if} $a \ge N_{1\star}$ and $2a \ge L$ \textbf{then} $B \gets B \cup \{(a,0) \}$.

Let $a, b \gets \ComputeAB(2)$.
\tcp*{compute $(a,b)$ under $b \ge 2$ if feasible}

\textbf{if} $a + b \ge N_{1\star}$ and $2a + b \ge L$ and $a \ge 0$ \textbf{then}
$B \gets B \cup \{(a,b)\}$.

\Return $\max \{\textsc{ComputeProbabilityViaFFT}(a, b, \Delta): \ (a,b) \in B \}$.
}

\caption{\textsc{ComputePmax}($\NN, \tau_0$)}
\label{alg:p_max}
\end{algorithm}

To find the lexicographically largest element $(a,b)$,
we start with feasible potential-outcome count vector $\vv$ and then convert the constraints on $\vv$ to constraints on $(a,b)$.
A potential-outcome count $\vv \in \mathbb{Z}_{\ge 0}^{\{0,1\}^2 }$ is feasible if and only if it satisfies:
\begin{enumerate}
\item Compatible with the observed data: 
\begin{align}
\begin{split}
& v_{11} + v_{10} \ge N_{11}, \quad 
v_{01} + v_{00} \ge N_{01}, \quad 
v_{11} + v_{01} \ge N_{10}, \quad 
v_{10} + v_{00} \ge N_{00}; \\    
& v_{11} \le N_{11} + N_{10}, \quad 
v_{10} \le N_{11} + N_{00}, \quad 
v_{01} \le N_{01} + N_{10}, \quad 
v_{00} \le N_{01} + N_{00}.
\end{split} \label{eqn:constraints_in_v_observe_data}
\end{align}

\item Compatible with the ATE $\tau_0$ and $n$:
\begin{align}
v_{10} - v_{01} = n\tau_0, \quad 
v_{11} + v_{10} + v_{01} + v_{00} = n.
\label{eqn:constraints_in_v_ate}
\end{align}
\end{enumerate}

Recall that $a \defeq v_{11}$ and $b \defeq v_{10} + v_{01}$. By constraints in Equation \eqref{eqn:constraints_in_v_ate}, we can express $\vv$ in terms of $a, b$:
\begin{align}
v_{11} = a, \quad
v_{10} = \frac{b + n\tau_0}{2}, \quad 
v_{01} = \frac{b - n\tau_0}{2}, \quad
v_{00} = n - a - b.
\label{eqn:v_express_in_ab}
\end{align}

Next, we convert the constraints in Equation \eqref{eqn:constraints_in_v_observe_data} to constraints in $a,b$.
Let 
\begin{align}
\begin{split}
& L \defeq \max\{2N_{11} - n \tau_0, 2N_{10}+ n \tau_0 \}, \quad
U \defeq \min\{2n - n\tau_0 - 2N_{01}, 2n + n\tau_0 - 2N_{00} \}, \\
& U_b \defeq \min \{2N_{11} + 2N_{00} - n\tau_0, \ 2N_{01} + 2N_{10} + n\tau_0 \}, \quad
N_{1\star} \defeq N_{11} + N_{10}.    
\end{split} \label{eqn:additional_parameters_for_ab}
\end{align}
Then, the constraints in Equation \eqref{eqn:constraints_in_v_observe_data} are equivalent to:
\begin{align}
\begin{split}
& L \le 2a + b \le U, \quad 
N_{1 \star} \le a + b \le n, \\
& 0 \le a \le N_{1 \star}, \quad \abs{n \tau_0} \le b \le U_b,
\quad b \equiv n\tau_0 \pmod 2.
\end{split}
\label{eqn:constraints_in_ab}
\end{align}

These constraints imply an upper bound for $a$, and given a fixed $a$, a corresponding upper bound for $b$:
\begin{align*}
& a \le \min \left\{\frac{1}{2}(U - \abs{n\tau_0}), \ n - \abs{n\tau_0}, \ N_{1\star}, \ U - N_{1\star} \right\},
\\
& b^{\max}(a) \le \max \left\{b \le \min \left\{ U - 2a, \ n - a, \ 
U_b \right\}: \ b \equiv n \tau_0 \pmod 2 \right\}.  
\end{align*}
We start with the largest possible $a$ and $b^{\max}(a)$ and then locally adjust them to satisfy all the constraints in Equation \eqref{eqn:constraints_in_ab}.

We use a helper local function \textsc{ComputeAB} in Algorithm \ref{alg:p_max} to identify the lexicographically largest element of $\mathcal{A}_0$ and  $\mathcal{A}_1$.
The helper function takes an additional input, $L_b \ge \abs{n\tau_0}$, as an additional lower bound on $b$. We set $L_b = \abs{n\tau_0}$ for $\tau_0 \neq 0$ and set $L_b = 2$ for $\tau_0 = 0$. 
Line 3 of the algorithm ensures that all upper bounds in Equation \eqref{eqn:constraints_in_ab} are satisfied.
The only constraints that may be violated are the following lower bounds:
\begin{align}
2a + b \ge L, \quad a + b \ge N_{1\star}, \quad 
b \ge \max\{L_b, \abs{n \tau_0}\}.
\label{eqn:ab_lower}
\end{align}
We show that, given that the set $\mathcal{A} \cap \{(a,b): b \ge L_b\}$ is nonempty, whenever a lower bound is violated, it can be restored by a local adjustment of $(a,b)$, and the helper function returns the lexicographically largest element of the set.

\begin{lemma}
Fix an observed table $\NN$ and $\tau_0 \in \mathcal{T}$, and let $L_b \ge \abs{n\tau_0}$ have the same parity as $\abs{n\tau_0}$.
If $\mathcal{A}' \defeq \mathcal{A} \cap \{(a,b) : b \ge L_b\}$ is nonempty, then the helper function \textsc{ComputeAB}($L_b$) in Algorithm \ref{alg:p_max}, returns the lexicographically largest element of $\mathcal{A}'$.
\label{lem:ab_search}
\end{lemma}

\begin{proof}
Let $(a^{\max},b^{\max})$ be computed in line 3 of Algorithm \ref{alg:p_max}. We consider three cases depending on the value $b^{\max}$.

Suppose $b^{\max} = U_b$, which has the same parity as $n\tau_0$. Since $\mathcal{A}' \neq \emptyset$, we have $b^{\max} = U_b \ge L_b$. 
We claim that all the other lower bounds in Equation \eqref{eqn:ab_lower} must also be satisfied at $(a^{\max}, b^{\max})$. 
Assume by contradiction, they are satisfied at some $a < a^{\max}$ and $b = b^{\max}(a)$. In this case, $b = U_b = b^{\max}$, and thus $2a+b < 2a^{\max} + b^{\max}$ and $a+b < a^{\max} + b^{\max}$, which is a contradiction.

Suppose $b^{\max} = U - 2a^{\max}$, which has the same parity as $n \tau_0$. Since $\mathcal{A}' \neq \emptyset$, we have $2a^{\max} + b^{\max} = U \ge L$.
In addition, by our choice of $a^{\max}$, we have
\begin{align*}
& a^{\max} + b^{\max} = U - a^{\max} \ge U - (U - N_{1\star}) = N_{1\star}, \\
& b^{\max} = U - 2a^{\max} \ge U - 2 \cdot \frac{1}{2} (U - L_b) = L_b.
\end{align*}
Thus, all inequalities in Equation \eqref{eqn:ab_lower} are satisfied.

Suppose $b^{\max} = n-a^{\max}$ (if $n-a^{\max} \equiv n\tau_0 \pmod 2$) or $n-a^{\max}-1$ (if $n-a^{\max} \not\equiv n\tau_0 \pmod 2$). 
In the first case,
\begin{align*}
& a^{\max} + b^{\max} = n \ge N_{1\star}, \\
& b^{\max} = n - a^{\max} \ge n - (n - L_b) \ge L_b.
\end{align*}
In addition, $2a+b \ge L$ must be satisfied by $(a^{\max}, b^{\max})$; otherwise, since it cannot be satisfied by any $a < a^{\max}$, the set $\mathcal{A}'$ is empty, which is a contradiction.
In the second case, $a^{\max} \ge n - L_b + 1$, and thus $b \ge L_b$; in addition, if $a+b \ge N_{1\star}$ is not satisfied by $(a^{\max}, b^{\max})$, it will be satisfied by $(a^{\max}-1, b^{\max}+2)$.
Finally, for the same reason as the first case, $2a+b \ge L$ must be satisfied $(a^{\max}, b^{\max})$ and thus by $(a^{\max}-1, b^{\max}+2)$.
\end{proof}

When $\tau_0 \neq 0$, the set $\mathcal{A}_0$ is empty. Lemma \ref{lem:ab_search}, together with Corollary \ref{coro:pmax_tau_0}, immediately implies the correctness of Algorithm \ref{alg:p_max} for $\tau_0 \neq 0$.

\begin{corollary}
Fix the observed data $\NN$ and $\tau_0 \in \mathcal{T} \setminus \{0\}$. Then, Algorithm \ref{alg:p_max} \textsc{ComputePmax}($\NN, \tau_0$) returns the value $p_{\max}(\tau_0)$ in time $O(1)$.
\label{coro:find_ab_1}
\end{corollary}

It remains to consider the case where $\tau_0 = 0$.

\begin{lemma}
Fix the observed data $\NN$, and assume $0 \in \mathcal{T}$. Then, Algorithm \ref{alg:p_max} \textsc{ComputePmax}($\NN, 0$) returns $p_{\max}(0)$  in time $O(1)$.
\label{lem:alg_max_p_tau_0}
\end{lemma}

\begin{proof}

We claim that Lines 12 -- 13 finds the lexicographically largest $(a,b)$ in $\mathcal{A}_0$, if $\mathcal{A}_0$ is nonempty.
Line 12 finds the maximum $a$ satisfying that 
\[
2a+b \le U, a+b \le n, \  a \le N_{1\star}, \  a \le U - N_{1\star}.
\]
Line 13 checks if the following lower bounds are also satisfied:
\[
a + b \ge N_{1\star} \ge 0,\  2a + b \ge L.
\]
If yes, then $(a,0)$ is the lexicographically largest element in $\mathcal{A}_0$.

We then claim that Lines 14 -- 15 finds the lexicographically largest $(a,b)$ in $\mathcal{A}_1$, if $\mathcal{A}_1$ is nonempty.
If $\mathcal{A}_1$ is nonempty, by Lemma \ref{lem:ab_search}, line 14 finds the lexicographically largest element in $\mathcal{A}_1$. 
If $\mathcal{A}_1$ is empty, by Lemma \ref{lem:ab_search}, $(a,b)$ cannot pass the condition in line 15.

By Corollary \ref{coro:pmax_tau_0}, \textsc{ComputePmax}($\NN, 0$) returns $p_{\max}(0)$.
\end{proof}

\subsubsection{Computing $p(a,b; \Delta)$ by the Fast Fourier Transform}

Recall that $p(a,b; \Delta) = \Pr(|S_{a,b}| \ge \Delta)$, where $S_{a,b}$ is a sum of independent Rademacher and half-Rademacher random variables.
Note that $S_{a,b}$ is a special case of the Poisson binomial distribution shifted by its expectation.
\citet{hong2013computing} develops an efficient algorithm that computes the cumulative distribution function (CDF) and thus the tails of the Poisson binomial distribution using the Fast Fourier Transform (FFT) \citep{cooley1965algorithm} (also see Chapter 30 of \citep{thomas2009introduction}).
We adopt this method\footnote{In this paper, we measure the runtime of the FFT under the unit-cost RAM model for simplicity. For discussions of practical floating-point implementations, bit complexity, and numerical stability, we refer readers to \citet{hong2013computing}.}.

\begin{lemma}
\label{lem:Pa_fft}
For any $(a,b) \in \mathbb{Z}_{\ge 0}^2$ and $\Delta > 0$, there exists an algorithm \\ \textsc{ComputeProbabilityViaFFT}($a,b,\Delta$) that computes $p(a,b;\Delta)$ in time 
$O(n \log n)$.
\end{lemma}

For completeness, we include a proof of Lemma \ref{lem:Pa_fft} in Appendix \ref{sec:appendix_proofs}. 

Combining Corollary \ref{coro:pmax_tau_0}, Corollary \ref{coro:find_ab_1}, Lemma \ref{lem:alg_max_p_tau_0}, and Lemma \ref{lem:Pa_fft}, we prove Theorem \ref{thm:time} for the balanced Bernoulli design.

\section{Fast Exact Confidence Set Construction for Matched-Pairs Design}
\label{sec:match}

In this section, we extend our fast confidence set construction for the balanced Bernoulli design to the matched-pairs design. We assume that $n$ is even so that the units can be paired.

Under the matched-pairs design, individual units are not assigned independently. Instead, assignments are independent across pairs, with each pair assigned to $(1,0)$ or $(0,1)$ with equal probability $1/2$, where $1$ indicates treatment and $0$ indicates control. Thus, if each pair is treated as a single ``unit,'' the independence and symmetry of the balanced Bernoulli design are preserved.
Consequently, most of the monotonicity properties developed in Section \ref{sec:four_permutation} continue to apply, enabling fast construction of exact confidence sets.
We formalize this extension in this section.

For convenience, we reuse and overload certain notation from the balanced Bernoulli design.

\subsection{Setup and Notation}

Without loss of generality, assume the $n$ units are partitioned into $m \defeq n/2$ disjoint pairs $(1,2),(3,4),\ldots,(n-1,n)$, by relabeling the units \citep{bai2022inference}. 
For each pair $i \in [m]$, let $Z_i \in \{0,1\}$ denote the within-pair treatment assignment. If $Z_i = 1$, then unit $2i$ is assigned to treatment and unit $2i-1$ to control; if $Z_i = 0$, then unit $2i-1$ is assigned to treatment and unit $2i$ to control.
Thus, the design makes $\{Z_i\}_{i=1}^m$ independent random variables with $\Pr(Z_i=1)=\Pr(Z_i=0)=1/2$ for all $i$.

For binary potential outcomes, define the \emph{pair-level} potential-outcome differences:
\begin{align}
W_i(1)=Y_{2i}(1)-Y_{2i-1}(0),\qquad
W_i(0)=Y_{2i-1}(1)-Y_{2i}(0).
\label{eqn:match_pair_outcome}
\end{align}
Thus, $W_i(1),W_i(0)\in\{-1,0,1\}$. The observed within-pair difference is
\[
W_i^{\mathrm{obs}} = Z_i\, W_i(1) + (1-Z_i)\, W_i(0).
\]
Let $\WW = (W_1, \ldots, W_m)$ and $\ZZ = (Z_1, \ldots, Z_m)$.
Under this notation, the ATE equals
\[
\tau(\WW) =\frac{1}{n}\sum_{i=1}^{m}\bigl(W_i(1)+W_i(0)\bigr),
\]
and the test statistic using the HT estimator\footnote{Under the matched-pairs design, the HT estimator is the same as the difference-in-means estimator since the two treatment group sizes are always equal.} is
\[
T(\WW, \ZZ)=\frac{2}{n}\sum_{i=1}^{m} W_i^{\mathrm{obs}}
= \frac{2}{n} \left( \sum_{i:Z_i=1} W_i(1) + \sum_{i:Z_i=0} W_i(0) \right).
\]

The set of all possible values of the ATE $\tau$ consistent with the observed data $(\WW^{\obs}, \ZZ^{\obs})$ is given by
\begin{align*}
\mathcal{T}(\WW^{\obs}, \ZZ^{\obs}) = \left\{ \frac{S^{\obs} - m}{n}, \ \frac{S^{\obs} - m + 1}{n}, \ \ldots, \ \frac{S^{\obs} + m}{n}  \right\}, 
\quad \text{where } S^{\obs} \defeq \sum_{i=1}^m W_i^{\obs}.
\end{align*}
This set is the same as that in Equation~\eqref{eqn:tau_set}, since the choice of design does not affect the true ATE.

\subsubsection{Inverting Randomization Tests under Matched Pairs}
We now specialize the general construction in Section \ref{sec:randomization_test_inversion} to the matched-pairs design $\mathcal D_{\mathrm{MP}}$.  Fix a candidate $\tau_0 \in \mathcal T(\WW^{\mathrm{obs}},\ZZ^{\mathrm{obs}})$, and let
$\mathcal W(\tau_0\WW^{\mathrm{obs}},\ZZ^{\mathrm{obs}})$ denote the set of compatible pair-level potential-outcome tables
$\tilde \WW$ that (1) agree with the observed within-pair differences and (2) satisfy $\tau(\tilde \WW)=\tau_0$.

For each $\tilde \WW \in \mathcal W(\tau_0, \WW^{\obs}, \ZZ^{\obs})$, we test the sharp null
$H_0(\tilde \WW): \WW_{\mathrm{true}}=\tilde \WW$ using the (centered) HT statistic under matched pairs,
and define the randomization-test $p$-value
\[
p(\tilde \WW, \ZZ^{\obs}) \defeq 
 \Pr_{\tilde \ZZ \sim \mathcal D_{\mathrm{MP}}}\!\left(
|T(\tilde \WW,\tilde Z)-\tau(\tilde \WW)|
\ge
|T(\tilde \WW,\ZZ^{\obs})-\tau(\tilde \WW)|
\right).
\]
As in Section \ref{sec:randomization_test_inversion},  we invert the randomization tests to form the exact $(1-\alpha)$ confidence set
\[
\mathcal{I}_{\alpha} = \{\tau_0 \in \mathcal{T}(\WW^{\obs}, \ZZ^{\obs}): p_{\max}(\tau_0) \ge \alpha \},
\quad \text{where }
p_{\max}(\tau_0) \defeq \max_{\tilde\WW \in \mathcal{W}(\tau_0, \WW^{\obs}, \ZZ^{\obs})} p(\tilde \WW, \ZZ^{\obs})
\]

As in the balanced Bernoulli design, $p_{\max}$ is monotone in $\tau_0$ under the matched-pairs design. This monotonicity follows from the independence and symmetry of the assignment mechanism across unit pairs. Thus, we can construct the confidence set $\mathcal I_\alpha$ using a binary search procedure analogous to Algorithm~\ref{alg:bernoulli_balance}.

\begin{lemma}
Under the matched-pairs design, the function $p_{\max}$ is monotone nondecreasing on $\mathcal{T} \cap [-1, T^{\obs}]$ and monotone nonincreasing on $\mathcal{T} \cap [T^{\obs}, 1]$, where $T^{\obs} = T(\tilde \WW, \ZZ^{\obs})$ is the observed statistic and is the same for all compatible $\tilde \WW$.
As a corollary, we can construct the confidence set $\mathcal{I}_\alpha$ for any given $\alpha \in (0,1)$ by binary search that evaluates $p_{\max}$ at most $4 \log_2 n$ distinct values.
\label{lem:f_monotone_matching}
\end{lemma}

The proof follows the same argument as that of Lemma~\ref{lem:f_monotone} for the balanced Bernoulli design, and thus we omit it.

\subsection{Computing $p_{\max}(\tau_0)$ Using at Most Two Randomization Tests}

\subsubsection{Expressing $p_{\max}(\tau_0)$ Using Pair-Level Outcome and Observation Counts}

Similar to the balanced Bernoulli design, for binary outcomes we can summarize observed data and pair-level potential outcomes using observation counts and potential-outcome counts, thereby reducing the number of parameters.

Let $N_{wz}$ denote the observed \emph{within-pair} counts
\[
N_{wz} \defeq |\{i\in[m]: W_i^{\mathrm{obs}}=w,\ Z_i=z\}| ,
\qquad w\in\{-1,0,1\},\ z \in\{0,1\}.
\]

Each pair has exactly one missing pair-level potential outcome: if $Z_i=1$ then $W_i(1)=W_i^{\mathrm{obs}}$ is observed and $W_i(0)$ is missing; if $Z_i=0$ then $W_i(0)=W_i^{\mathrm{obs}}$ is observed and $W_i(1)$ is missing. Since both missing and observed values lie in $\{-1,0,1\}$, it is natural to parameterize a potential-outcome completion by counts of missing values conditional on the observed value.
For $w,u\in\{-1,0,1\}$, let $v_{w,u}$ be the number of pairs with $W_i^{\obs}=w$ whose missing pair-level potential outcome is imputed as $u$. 
Let 
\[
\NN = \{N_{wz}\} \in \mathbb{Z}_{\ge 0}^{\{-1,0,1\} \times \{0,1\} }, \qquad \vv = \{v_{w,u}\} \in \mathbb{Z}_{\ge 0}^{\{-1,0,1\}^2 }.
\]
Let 
\[
p(\vv, \NN) = p(\tilde\WW, \ZZ^{\obs}),
\]
for any $\tilde\WW$ whose associated potential-outcome count is $\vv$.

\subsubsection{Reduction to a Sum of Independent Weighted Rademacher Variables}

Similar to Lemma \ref{lem:objective_in_v}, under the matched-pairs design, we can write the centered test statistic as a sum of independent weighted Rademacher random variables.

For pair $i \in [m]$, define the pair contrast
\[
D_i \defeq W_i(1)-W_i(0)\in\{-2,-1,0,1,2\}.
\]
Let $\tilde \ZZ = (\tilde{Z}_1, \ldots, \tilde{Z}_m)$ be a random assignment under the matched-pairs design. Define $\xi_i=2\tilde Z_i-1 \in \{\pm 1\}$. The centered test statistic is
\begin{align*}
T(\tilde \WW,\tilde \ZZ)-\tau_0
&=
\frac{2}{n}\sum_{i=1}^m\bigl(\tilde Z_i W_i(1) + (1-\tilde Z_i)W_i(0)\bigr)
-
\frac{1}{n}\sum_{i=1}^m \bigl(W_i(1)+W_i(0)\bigr) \\
&=
\frac{1}{n}\sum_{i=1}^m (2\tilde Z_i-1)\,\bigl(W_i(1)-W_i(0)\bigr) \\
&=
\frac{1}{n}\sum_{i=1}^m \xi_i\, D_i. 
\end{align*}
Under the matched-pairs design, the $\xi_i$ are independent Rademacher random variables. 
Moreover, since $\xi_i D_i$ has the same distribution as $\xi_i |D_i|$, the randomization distribution depends on the completion only through the counts of pairs with $|D_i|=2$ and $|D_i|=1$.

Given a pair-level potential outcome table $\vv$, define
\begin{equation}\label{eq:m2m1}
m_2 \defeq v_{1,-1}+v_{-1,1},\qquad
m_1 \defeq v_{1,0}+v_{0,1}+v_{0,-1}+v_{-1,0}.
\end{equation}
Pairs counted by $m_2$ have $(W_i(1),W_i(0))\in\{(1,-1),(-1,1)\}$ so $|D_i|=2$, and pairs counted by $m_1$ have $(W_i(1),W_i(0))$ differing by $1$ so $|D_i|=1$. The remaining pairs have $|D_i|=0$ and do not affect the randomization distribution.

Therefore, 
\[
p(\vv, \NN)
= \Pr (|W_{m_2, m_1} | \ge \Delta), \qquad 
\Delta \defeq n\,|T^{\obs}-\tau_0|.
\]
The random variable $W_{m_2, m_1}$ has the same distribution as
\[
W_{m_2, m_1} \overset{d}{=}  \sum_{r=1}^{m_2} \eta_r \;+\; \sum_{s=1}^{m_1} \epsilon_s,
\]
where $\eta_r \sim \text{Unif}\{\pm 2\}$ and $\eps_s \sim \text{Unif}\{\pm 1\}$ are independent random variables. 
In particular, $p(\vv, \NN)$ depends on $\vv$ only through the pair $(m_2, m_1)$, and we define $p(m_2, m_1; \Delta) = p(\vv, \NN)$ where $m_2, m_1$ are defined as in Equation \eqref{eq:m2m1}.

The random variable $W_{m_2,m_1}$ has the same distribution as $2S_{a,b}$ in Lemma~\ref{lem:objective_in_v} for the balanced Bernoulli design when $(m_2,m_1)=(a,b)$. Consequently, the monotonicity properties established for $S_{a,b}$ carry over directly to $W_{m_2,m_1}$. In particular, when $m_1 \ge 1$, Lemma~\ref{lem:g_amax} implies that $p(\vv,\NN)$ is maximized at the lexicographically largest feasible pair $(m_2,m_1)$. When $m_1 = 0$, Lemma~\ref{lem:monotone_pab_b0} shows that $p(\vv,\NN)$ is maximized at the largest attainable value of $m_2$. 

As a corollary, we obtain the following characterization of the parameter pairs $(m_2,m_1)$ that achieve the maximum of $p(m_2,m_1;\Delta)$, and hence of $p_{\max}(\tau_0)$.
Its proof is the same as that of Corollary \ref{coro:pmax_tau_0} and thus we omit it.

\begin{lemma}
Assume $\Delta > 0$ and $\tau_0 \in \mathcal{T}$, and fix the observed table $\NN$. 
Let $\mathcal{M}$ be the feasible set of $(m_2, m_1)$.
Define
\[
\mathcal{M}_0 \defeq \{(m_2, m_1) \in \mathcal{M} : m_1 = 0\},
\qquad
\mathcal{M}_1 \defeq \{(m_2, m_1) \in \mathcal{M} : m_1 \ge 1\}.
\]
Then, at least one of these sets is nonempty.
The maximum value $p_{\max}(\tau_0)$ is given by
\[
p_{\max}(\tau_0)
=
\max\!\left\{
\max\{ p(m_2, m_1;\Delta) : (m_2, m_1) \in \mathcal{M}_0 \},
\;
\max\{ p(m_2, m_1;\Delta) : (m_2, m_1) \in \mathcal{M}_1 \}
\right\},
\]
with the convention that if one of the sets is empty, the corresponding maximum is omitted.
Moreover, in each of the sets $\mathcal{M}_0$ and $\mathcal{M}_1$, the corresponding maximum is attained at the lexicographically largest element of the set.
\label{lem:pmax_tau_0_matching}    
\end{lemma}

Lemma \ref{lem:pmax_tau_0_matching} implies that at most two randomization tests are needed to compute $p_{\max}(\tau_0)$ for a given $\tau_0$. Together with Lemma \ref{lem:f_monotone_matching}, we prove Theorem \ref{thm:upper} for the matched-pairs design.

\subsubsection{Finding the Lexicographically Largest Feasible $(m_2, m_1)$ via Integer Linear Programming}

A feasible pair-level potential-outcome count vector $\vv = \{v_{w,u}\} \in \mathbb{Z}_{\ge 0}^{\{-1,0,1\}^2 } $ needs to satisfy the following constraints:
\begin{enumerate}
\item Compatible with the observed data\footnote{Observe that this implies $\sum_{w,u} v_{w,u} = m$.}:
\begin{align}
\sum_{u\in\{-1,0,1\}} v_{w,u}=N_w, \quad \text{where } N_w \defeq N_{w,1} + N_{w,0}, \quad \text{for } w\in\{-1,0,1\}.
\label{eqn:constraint_matching_observe}
\end{align}

\item Compatible with the ATE $\tau_0$:
\begin{align}
\sum_{w, u \in \{-1,0,1\}} u \, v_{w,u} = n\tau_0 - S^{\obs}.
\label{eqn:constraint_matching_ate}
\end{align}

\end{enumerate}

That is, 
\[
\mathcal{M} = \left\{(m_2, m_1): \text{ exists  $\vv$ satisfying Equations \eqref{eqn:constraint_matching_observe} and  \eqref{eqn:constraint_matching_ate}} \right\}.
\]
Let 
\[
\mathcal{V} \defeq \{\vv \in \mathbb{Z}_{\ge 0}^{\{-1,0,1\}^2 }: \vv \text{ satisfies Equations \eqref{eqn:constraint_matching_observe} and  \eqref{eqn:constraint_matching_ate}} \}
\]
denote the feasible set of $\vv$.

We find the lexicographically largest $(m_2, m_1)$ in $\mathcal{M}$ by solving two integer linear programs (ILPs).
Different from the balanced Bernoulli design, here, 
we do not have a simple characterization of $\mathcal{M}$ and a simple search of the lexicographically largest $(m_2, m_1)$, since the outcome table $\vv$ involves 9 parameters.

For the set $\mathcal{M}_0$, where $m_1 = 0$, we find the maximum feasible $m_2$ by solving
\[
\max \ v_{1,-1} + v_{-1,1}, \quad 
\text{subject to } \vv \in \mathcal{V} \text{ and } 
v_{1,0}+v_{0,1}+v_{0,-1}+v_{-1,0} = 0.
\]
For the set $\mathcal{M}_1$, where $m_1 \ge 1$, we first find the maximum feasible $m_2$ by solving
\[
\max \ v_{1,-1} + v_{-1,1}, \quad 
\text{subject to } \vv \in \mathcal{V} \text{ and } 
v_{1,0}+v_{0,1}+v_{0,-1}+v_{-1,0} \ge 1.
\]
Let $v_{1,-1}^*, v_{-1,1}^*$ be the optimal solution. Then, we find the maximum feasible $m_1 = v_{1,0}+v_{0,1}+v_{0,-1}+v_{-1,0}$ under the optimal $v_{1,-1}^*, v_{-1,1}^*$ by solving
\[
\max \ v_{1,0}+v_{0,1}+v_{0,-1}+v_{-1,0}, \quad 
\text{subject to } \vv \in \mathcal{V} \text{ and } 
v_{1,-1} = v_{1,-1}^*, \ v_{-1,1} = v_{-1,1}^*.
\]

Observe that all the above three ILPs have $9$ variables and a constant number of linear constraints; in addition, all the coefficients in front of the variables are in $\{-1,0,1\}$.
\citet{jansen2023integer} show that such ILP can be solved in $O(1)$ arithmetic operations\footnote{More generally, \citet{jansen2023integer} consider the integer linear program $\max\{\cc^\top \xx : \AA \xx = \bb,\ \xx \in \mathbb{Z}_{\ge 0}^n\}$ where $\AA \in \mathbb{Z}^{m \times n}$, $\bb \in \mathbb{Z}^m$, and $\cc \in \mathbb{Z}^n$. They develop an algorithm using $O(\sqrt{m}\,\Delta)^{2m} + O(nm)$ arithmetic operations, where $\Delta$ is the maximum absolute entry of $\AA$. In our setting, $m$, $n$, and $\Delta$ are constants.}.

Finally, we can use the same FFT algorithm \textsc{ComputeProbabilityViaFFT}($m_2, m_1; \Delta/2$) to compute the $p(m_2, m_1; \Delta)$ in $O(n \log n)$ time.

Therefore, we prove Theorem \ref{thm:time} for the matched-pairs design.

\section{Simulations}
\label{sec:simulation}

We first compare our Algorithm~\ref{alg:bernoulli_balance} to an exact brute-force procedure for small values of $n$ in Table \ref{tab:small_instances}. 
The brute-force method enumerates all potential-outcome count vectors compatible with the observed data $\NN$, and therefore requires $\Theta(n^4)$ randomization tests, making it computationally infeasible beyond small sample sizes.
The number of randomization tests needed by our Algorithm \ref{alg:bernoulli_balance} is much smaller than that needed by the brute-force.

\begin{table}[!ht]
\centering
\caption{$95\%$ confidence intervals under the balanced Bernoulli design. 
The first column reports the observed count vector $\NN = (N_{11}, N_{01}, N_{10}, N_{00})$. 
The second column gives the confidence interval $[\min \mathcal I_\alpha, \max \mathcal I_\alpha]$, scaled by $n$. 
The third and fourth columns report the number of randomization tests required by a brute-force inversion and by our fast construction in Algorithm~\ref{alg:bernoulli_balance}, respectively.}
\label{tab:small_instances}
\begin{tabular}{|c|c|c|c|}
\hline
$\NN$ & Confidence interval & Brute-Force &  Algorithm \ref{alg:bernoulli_balance} \\
\hline
$(2,6,8,0)$ & $[-14, 0]$ & 189 & 7 \\
$(6,4,4,6)$ & $[-7, 12]$ & 1225 & 8 \\
$(8,4,5,7)$ & $[-7, 15]$ & 2160 & 8 \\
$(10,13,15,12)$ & $[-27, 11]$ & 32032 & 9 \\
\hline
\end{tabular}
\end{table}

We then compare Algorithm~\ref{alg:bernoulli_balance} with Wald confidence intervals constructed using an asymptotic normal approximation in Table~\ref{tab:large_instances}. Following \citet{aronow2023fast}, we study two settings: (1) $50\%$ of units have potential outcomes $(1,1)$ and the remaining $50\%$ have $(0,0)$; and (2) $8\%$ of units have potential outcomes $(1,1)$ while the remaining $92\%$ have $(0,0)$. In both settings, the ATE equals $0$.

Algorithm \ref{alg:bernoulli_balance} inverts finite-sample randomization tests and therefore attains coverage uniformly for all potential outcome settings. In the sparse-outcome setting (the second setting), the observed data are less informative about the missing potential outcomes and the randomization distribution is highly discrete and non-normal, so Wald's asymptotic normal approximation can undercover. The resulting exact inversion is correspondingly more conservative, leading to wider intervals (and, in our simulations, occasional over-coverage).
The conservativeness comes from the Bernoulli design, not from our fast construction algorithm.

\begin{table}[!ht]
\centering
\caption{$95\%$ confidence intervals under the balanced Bernoulli design for two settings. Coverage (Cov.) and median interval width (Med. Width) are reported for Algorithm~\ref{alg:bernoulli_balance} and Wald intervals. For each $n$, we run $10{,}000$ independent simulations.}
\label{tab:large_instances}
\begin{tabular}{|c|cc|cc|cc|cc|}
\hline
& \multicolumn{4}{c|}{\textbf{First Setting}} & \multicolumn{4}{c|}{\textbf{Second Setting}} \\
\hline
& \multicolumn{2}{c|}{Exact} & \multicolumn{2}{c|}{Wald}
& \multicolumn{2}{c|}{Exact} & \multicolumn{2}{c|}{Wald} \\
\hline
$n$
& Cov. & Med. Width & Cov. & Med. Width
& Cov. & Med. Width & Cov. & Med. Width \\
\hline
50  & 0.98 & 0.76 & 0.96 & 0.79 & 1.00 & 0.54 & 0.88 & 0.31 \\
100 & 0.96 & 0.57 & 0.93 & 0.56 & 1.00 & 0.41 & 0.92 & 0.22 \\
200 & 0.98 & 0.42 & 0.95 & 0.39 & 1.00 & 0.30 & 0.92 & 0.16 \\
1000 & 0.97 & 0.19 & 0.95 & 0.18 & 1.00 & 0.14 & 0.94 & 0.07 \\
\hline
\end{tabular}
\end{table}

Finally, we compare our fast construction for the matched-pairs design with the corresponding Wald confidence intervals in Table \ref{tab:match}.
We study two settings: (1) $50\%$ of unit pairs have pair-level potential-outcome differences (as defined in Equation \eqref{eqn:match_pair_outcome}) $(1,1)$ and the remaining $50\%$ have $(1,0)$; and (2) $92\%$ of unit pairs have pair-level potential-outcome differences $(1,1)$ and the remaining $8\%$ have $(1,0)$ (in $92\%$ pairs, the two pair-level potential outcomes are equal). In the first setting, the ATE is $0.75$; in the second setting, it is $0.96$.
Our exact construction guarantees coverage for all $n$, and the interval widths are comparable to those of Wald intervals.

\begin{table}[!ht]
\centering
\caption{$95\%$ confidence intervals under the matched-pairs design for two settings. For each $n$, we run $10{,}000$ independent simulations. In the first setting with $n=50$, since the number of pairs is odd, we set $13$ pairs with pair-level potential-outcome differences $(1,1)$ and $12$ pairs with $(1,0)$.}
\label{tab:match}
\begin{tabular}{|c|cc|cc|cc|cc|}
\hline
& \multicolumn{4}{c|}{\textbf{First Setting}} & \multicolumn{4}{c|}{\textbf{Second Setting}} \\
\hline
& \multicolumn{2}{c|}{Exact} & \multicolumn{2}{c|}{Wald}
& \multicolumn{2}{c|}{Exact} & \multicolumn{2}{c|}{Wald} \\
\hline
$n$
& Cov. & Med. Width & Cov. & Med. Width
& Cov. & Med. Width & Cov. & Med. Width \\
\hline
50  & 1.00 & 0.38 & 0.98 & 0.34 & 1.00 & 0.22 & 0.76 & 0.16 \\
100 & 0.99 & 0.25 & 0.98 & 0.25 & 1.00 & 0.13 & 0.94 & 0.11 \\
200 & 0.99 & 0.18 & 0.98 & 0.17 & 0.99 & 0.08 & 0.96 & 0.08 \\
1000 & 0.99 & 0.08 & 0.98 & 0.08 & 0.99 & 0.03 & 0.98 & 0.03 \\
\hline
\end{tabular}
\end{table}

\paragraph{Runtime.} All simulations were conducted on a 2020 iMac equipped with a 3.8 GHz 8-core Intel Core i7 processor and 8 GB of RAM.
For 10,000 sequential simulations with $n=1000$, the total runtime is under 1.5 minutes for the balanced Bernoulli design and under 2 hours for the matched-pairs design. The increased computational cost in the matched-pairs design is because each simulation needs to solve two integer linear programs with nine variables.

\bibliographystyle{plainnat}
\bibliography{ref}

\appendix

\section{Missing Proofs}
\label{sec:appendix_proofs}

In this section, we include all the missing details and proofs from the main text.

\subsection{Missing Proofs in Section \ref{sec:binary_search}}

We first prove the monotonicity of \(p_{\max}\) under the balanced Bernoulli design.
The proof relies on the symmetry of the design between the two treatment groups, that is, each unit is assigned to either group with equal probability.

\begin{proof}[Proof of Lemma \ref{lem:f_monotone}]

Given the observed data $(\YY^{\obs}, \ZZ^{\obs})$, let $\mathcal{T} = \mathcal{T}(\YY^{\obs}, \ZZ^{\obs})$, and let $T^{\obs}$ be the test statistic (the HT estimator) under the observed data.

We first consider $\tau_0 \in \mathcal{T} \cap [-1, T^{\obs}]$. The argument for $\tau_0 \in \mathcal{T} \cap [T^{\obs}, 1]$ is similar.

Consider a fixed $\tau_0$ and $\tau_1 = \tau_0 + 1/n$ both are in $\mathcal{T} \cap [-1, T^{\obs}]$.
Both $\mathcal{Y}(\tau_0, \YY^{\obs}, \ZZ^{\obs})$ and $\mathcal{Y}(\tau_1, \YY^{\obs}, \ZZ^{\obs})$ are nonempty.
Suppose $p_{\max}(\tau_0)$ is attained by some $\tilde{\YY}_0 \in \mathcal{Y}(\tau_0, \YY^{\obs}, \ZZ^{\obs})$. 
Consider the following $\tilde{\YY}_1$: if there is an imputed potential outcome under \emph{treatment} is $0$, we turn it to $1$; otherwise, there must be an imputed potential outcome under \emph{control} is $1$ (since $\tau_1 = \tau_0 + 1/n \in \mathcal{T}$), and we turn it to $0$. 
Then, $\tilde{\YY}_1 \in \mathcal{Y}(\tau_1, \YY^{\obs}, \ZZ^{\obs})$.

We will show that 
\begin{align}
\Pr_{\tilde{\ZZ}} \left( \abs{T (\tilde{\YY}_1, \tilde{\ZZ}) - \tau_1} \ge \abs{T^{\obs} - \tau_1} \right) \ge \Pr_{\tilde{\ZZ}} \left( \abs{T (\tilde{\YY}_0, \tilde{\ZZ}) - \tau_0} \ge \abs{T^{\obs} - \tau_0} \right) ,
\label{eqn:bernoulli_goal}
\end{align}
which implies $p_{\max}(\tau_1) \ge p_{\max}(\tau_0)$.
For $\tau_0, \tau_1 \in [-1, T^{\obs}]$, we have that for $j = 0,1$,
\[
\abs{T (\tilde{\YY}_j, \tilde{\ZZ}) - \tau_j} \ge \abs{T^{\obs} - \tau_j}
\iff T (\tilde{\YY}_j, \tilde{\ZZ}) \ge T^{\obs}
\text{ or } 
T (\tilde{\YY}_j, \tilde{\ZZ}) \le 2 \tau_j - T^{\obs}
\]

Consider an arbitrary assignment $\ZZ' \in \{0,1\}^n$.
Since
\[
\tilde{Y}_{0,i}(1) \le \tilde{Y}_{1,i}(1), \ 
\tilde{Y}_{0,i}(0) \ge \tilde{Y}_{1,i}(0), \quad \text{for all } i
\]
we have 
\[
T(\tilde{\YY}_0, \ZZ') \le T(\tilde{\YY}_1, \ZZ').
\]
Thus,
\[
\Pr_{\tilde{\ZZ}} \left( T(\tilde{\YY}_1, \tilde{\ZZ}) \ge T^{\obs} \right) \ge \Pr_{\tilde{\ZZ}} \left( T(\tilde{\YY}_0, \tilde{\ZZ}) \ge T^{\obs} \right).
\]
In addition, 
\begin{align*}
& T(\tilde{\YY}_1, \ZZ') \le T(\tilde{\YY}_0, \ZZ') + \frac{2}{n}
= T(\tilde{\YY}_0, \ZZ') + 2(\tau_1 - \tau_0)  \\
\implies & T(\tilde{\YY}_1, \ZZ') - 2\tau_1 \le T(\tilde{\YY}_0, \ZZ') - 2\tau_0.
\end{align*}
Thus, 
\[
\Pr_{\tilde{\ZZ}} \left( T(\tilde{\YY}_1, \tilde{\ZZ}) \le 2\tau_1 - T^{\obs} \right) \ge \Pr_{\tilde{\ZZ}} \left( T(\tilde{\YY}_0, \tilde{\ZZ}) \le 2\tau_0 - T^{\obs} \right).
\]
Therefore, Equation \eqref{eqn:bernoulli_goal} holds and $p_{\max}(\tau_1) \ge p_{\max}(\tau_0)$.

Next, we consider $\tau_0 > T^{\obs}$. 
Consider a fixed $\tau_0$ and $\tau_1 = \tau_0 - 1/n$ both are in $\mathcal{T} \cap [T^{\obs}, 1]$.
Again, suppose $p_{\max}(\tau_0)$ is attained at some $\tilde{\YY}_0$.
Consider $\tilde{\YY}_1$ that turns one imputed potential under treatment from $1$ to $0$ or turns one imputed potential under control from $0$ to $1$ whichever is possible. Then, $\tilde{\YY}_1 \in \mathcal{\YY}(\tau_1, \YY^{\obs}, \ZZ^{\obs})$.

Consider an arbitrary $\ZZ' \in \{0,1\}^n$.
\begin{align*}
T(\tilde{\YY}_0, \ZZ') - \frac{2}{n} \le T(\tilde{\YY}_1, \ZZ') \le  T(\tilde{\YY}_0, \ZZ').
\end{align*}
Thus, 
\begin{align*}
& \Pr_{\tilde{\ZZ}} \left( T(\tilde{\YY}_1, \tilde{\ZZ}) \le T^{\obs} \right) \ge \Pr_{\tilde{\ZZ}} \left( T(\tilde{\YY}_0, \tilde{\ZZ}) \le T^{\obs} \right), \\
& \Pr_{\tilde{\ZZ}} \left( T(\tilde{\YY}_1, \tilde{\ZZ}) \ge 2\tau_1 - T^{\obs} \right)
\ge \Pr_{\tilde{\ZZ}} \left( T(\tilde{\YY}_0, \tilde{\ZZ}) \ge 2\tau_0- T^{\obs} \right).
\end{align*}
Therefore, $p_{\max}(\tau_1) \ge p_{\max}(\tau_0)$.

Combining the above two cases, we prove the statement.
\end{proof}

Next, we present a detailed binary search algorithm that proves Corollary \ref{coro:binary_search}.

\begin{algorithm}[t]
	\SetAlgoNoLine
	\KwIn{Observed data $\YY^{\obs}, \ZZ^{\obs}$ and confidence level $\alpha \in (0,1)$}
	\KwOut{A $(1-\alpha)$ confidence set $\mathcal{I}_{\alpha}$}
    Write $\mathcal T (\YY^{\obs}, \ZZ^{\obs})$ in Equation \eqref{eqn:tau_set} as $\mathcal{T} (\YY^{\obs}, \ZZ^{\obs}) =\{\tau_{(1)}<\tau_{(2)}<\cdots<\tau_{(K)}\}$.

    Let $T^{\obs} \gets T(\YY^{\obs}, \ZZ^{\obs})$.

\tcp{Corner cases}

\If{
$T^{\obs} < \tau_{(1)}$
}{
Find $k_{\max} \gets \max\{1 \le k \le K: p_{\max}(\tau_{(k)}) \ge \alpha \}$ by binary search over $\{1,\dots,K\}$ using function $f$ in Algorithm \ref{alg:bernoulli_balance}, and return $\mathcal{I}_{\alpha} = \{\tau_{(1)}, \ldots, \tau_{(k_{\max})}\}$.
}
\If{
$T^{\obs} > \tau_{(K)}$
}{
Find $k_{\min} \gets \min\{ 1 \le k \le K: p_{\max}(\tau_{(k)}) \ge \alpha \}$ by binary search over $\{1,\dots,K\}$ using function $f$ in Algorithm \ref{alg:bernoulli_balance}, and return $\mathcal{I}_{\alpha} =  \{\tau_{(k_{\min})}, \ldots, \tau_{(K)}\}$.
}

\tcp{General case}

Let 
\[
k_- \gets \max \{1 \le k \le K: \tau_{(k)} \le T^{\obs} \}, \quad 
k_+ \gets \min \{1 \le k \le K: \tau_{(k)} \ge T^{\obs} \}.
\]

Find $k_{\min} \gets \min\{1 \le k \le k_-: p_{\max}(\tau_{(k)}) \ge \alpha \}$ by binary search over $\{1,\dots, k_-\}$ using function $f$ in Algorithm \ref{alg:bernoulli_balance}. If none, let $k_{\min} \gets k_+$.

Find $k_{\max} \gets \max\{k_+ \le k \le K: p_{\max}(\tau_{(k)}) \ge \alpha \}$ by binary search over $\{k_+,\dots,K\}$ using function $f$ in Algorithm \ref{alg:bernoulli_balance}. If none, let $k_{\max} \gets k_-$.

\Return $\mathcal{I}_{\alpha} = \{\tau_{(k_{\min})}, \ldots, \tau_{(k_{\max})} \}$.

	\caption{Binary Search}
	\label{alg:binary_search}
\end{algorithm}

We prove the lower bound in Lemma \ref{lem:lower}.

\begin{proof}[Proof of Lemma \ref{lem:lower}]

The feasible grid of ATE values $\mathcal{T} = \mathcal{T}(\YY^{\obs}, \ZZ^{\obs})$ in Equation \eqref{eqn:tau_set} has $K = n+1$ values. Write $\mathcal{T} = \{\tau_{(1)} < \ldots < \tau_{(K)} \}$
Define the acceptance indicator
\[
b(k) \defeq \mathbbm{1}\{p_{\max}(\tau^{(k)})\ge \alpha\},\qquad k\in\{1,\dots,K\}.
\]
Then,
\[
\mathcal{I}_\alpha = \{\tau_{(k)}: \ b(k)=1\}.
\]

By Lemma 2.1, $p_{\max}(\tau)$ is monotone nondecreasing on $\mathcal{T} \cap[-1,T^{\mathrm{obs}}]$
and monotone nonincreasing on $\mathcal{T} \cap[T^{\mathrm{obs}},1]$.
Hence the sequence $(b(1),\dots,b(K))$ has the form
\[
0,\dots,0,\;1,\dots,1,\;0,\dots,0,
\]
i.e., the set $\{k: b(k)=1\}$ is a contiguous block of indices.

To prove a lower bound, it suffices to consider a special (easier) subproblem in which
$b$ is \emph{monotone}. For example, in the boundary case
$T^{\mathrm{obs}}<\tau^{(1)}$ (or symmetrically $T^{\mathrm{obs}}>\tau^{(K)}$), the monotonicity
of $p_{\max}$ implies that $b(1),\dots,b(K)$ is monotone.  Thus, there exists an index
$s\in\{1,\dots,K+1\}$ such that
\[
b(k)=
\begin{cases}
0, & k<s,\\
1, & k\ge s,
\end{cases}.
\]
with the convention $s=K+1$ corresponding to $b(k)\equiv 0$ and hence $\mathcal I_\alpha=\emptyset$.
Equivalently, in this subproblem, the confidence set is exactly the suffix
$\mathcal I_\alpha=\{\tau_{(s)}, \dots,\tau_{(K)}\}$, and there are $K+1$ possible answers
(one for each $s$).

Now consider any deterministic algorithm $\mathcal{A}$ that always identifies $\mathcal I_\alpha$
exactly. Suppose that $\mathcal{A}$ has access to an oracle which, on query $k$, returns the single bit
$b(k)$. Such an oracle needs at least one permutation test, and thus a lower bound on the number of oracle queries needed by $\mathcal{A}$ is also a lower bound on the number of permutation tests needed.

Suppose $\mathcal{A}$ makes at most $q$ (adaptive) oracle queries in the worst case.
Because $\mathcal{A}$ is deterministic, for each $s\in\{1,\dots,K+1\}$ the sequence of oracle
responses is a binary string of length $q$.  Hence $\mathcal{A}$ can induce at most $2^q$
distinct transcripts and therefore can distinguish at most $2^q$ different values of $s$.
But $\mathcal{A}$ must be correct for all $K+1$ possibilities.  Therefore, we must have
$2^q\ge K+1$, i.e.,
\[
q \ge \log_2(K+1) = \Omega(\log n).
\]
This proves that any deterministic algorithm that always identifies $\mathcal I_\alpha$
exactly needs $\Omega(\log n)$ permutation tests in the worst case.
\end{proof}

Finally, we prove the $n^{-1/2}$ convergence rate in Proposition \ref{prop:bern-hoeffding-length}. We prove it for general Bernoulli designs with marginal treatment assignment probability $p \in (0,1)$. This convergence rate only depends on how we construct the set $\mathcal{I}_\alpha$ in Equation \eqref{eqn:ci} and does not depend on how we implement the construction.

\begin{proof}[Proof of Proposition \ref{prop:bern-hoeffding-length}]
Fix any $\tau_0$ and any completion $\tilde \YY\in\mathcal{Y}(\tau_0, \YY^{\obs}, \ZZ^{\obs})$.
Let $\tilde{\ZZ} \sim \Ber(p)^n$ be a random assignment drawn under the Bernoulli design with marginal probability $p$.
Define
\[
X_i
\defeq \frac{\tilde Z_i\,\tilde Y_i(1)}{p}-\frac{(1-\tilde Z_i)\,\tilde Y_i(0)}{1-p}, \quad \text{for each } i \in [n],
\]
and $S \defeq \sum_{i=1}^n X_i$.
Random variables $X_1, \ldots, X_n$ are independent and satisfy
\[
X_i \in \{ \frac{1}{p}, \  - \frac{1}{1-p} \}, \quad \text{for each } i
\]
and satisfy
\[
T(\tilde{\YY}, \tilde{\ZZ}) = \frac{1}{n} \sum_{i=1}^n X_i, \quad 
\mathbb{E}[T (\tilde{\YY}, \tilde{\ZZ})] = \tau_0.
\]

Therefore, Hoeffding's inequality yields for any $t>0$,
\begin{align*}
\Pr \left( \abs{T(\tilde{\YY}, \tilde{\ZZ}) - \tau_0} \ge t \right)  
& = \Pr \left( \abs{S - \mathbb{E} [S]} \ge n t \right) \\
& \le 2 \exp \left( - \frac{2n^2 t^2}{n \cdot (1/p + 1/(1-p))^2} \right) \\
& = 2 \exp \left( - 2n t^2 p^2 (1-p)^2 \right).
\end{align*}
Consider $t = \abs{T^{\obs} - \tau_0}$. If 
\[
\abs{T^{\obs} - \tau_0} > \sqrt{\frac{\log (2/\alpha)}{2n p^2 (1-p)^2}},
\]
then $\Pr \left( \abs{T(\tilde{\YY}, \tilde{\ZZ}) - \tau_0} \ge \abs{T^{\obs} - \tau_0} \right) < \alpha$ for any completion $\tilde{\YY} \in \mathcal{Y}(\tau_0, \YY^{\obs}, \ZZ^{\obs})$. That is, $p_{\max}(\tau_0) < \alpha$ and thus $\tau_0 \notin \mathcal{I}_{\alpha}$.

Therefore any $\tau_0\in \mathcal{I}_\alpha$ must satisfy
\[
\abs{T^{\obs}-\tau_0}
\le
\sqrt{\frac{\log(2/\alpha)}{2n p^2(1-p)^2}}.
\]
This implies the length of the interval $[\max \mathcal{I}_\alpha, \ \min \mathcal{I}_\alpha]$ satisfies
\[
\max \mathcal{I}_\alpha - \min \mathcal{I}_\alpha
\le
2\sqrt{\frac{\log(2/\alpha)}{2n p^2(1-p)^2}}.
\]
Taking $p = 1/2$, we complete the proof of the proposition statement.
\end{proof}

\subsection{Missing Proofs in Section \ref{sec:four_permutation} }

Section \ref{sec:four_permutation} bounds the number of randomization tests needed to compute $p_{\max}(\tau_0)$ for a given $\tau_0 \in \mathcal{T}(\YY^{\obs}, \ZZ^{\obs})$.

Given a fixed $\tau_0$, Lemma \ref{lem:objective_in_v} expresses the centered test statistic as $S_{a,b}$, a sum of $a$ independent Rademacher and $b$ independent half-Rademacher random variables.
Lemma \ref{lem:unimodal} shows that $S_{a,b}$ is unimodal.

\begin{proof}[Proof of Lemma \ref{lem:unimodal}]
Recall from the proof of Lemma \ref{lem:objective_in_v},
$x_1 \sim \Bin(v_{11}, \frac{1}{2})$ and $x_2 + x_3 \sim \Bin(v_{10} + v_{01}, \frac{1}{2})$.
A binomial distribution is strongly unimodal, in the sense that its convolution with any unimodal discrete distribution remains unimodal \citep{keilson1971some}.
Thus, $2x_1 - \mathbb{E}[S_{a,b}]$ is unimodal (an affine transform of $x_1$), and $x_2+x_3$ is strongly unimodal.
It follows that the convolution of these two distributions, that is, the distribution of $S_{a,b} - \mathbb{E}[S_{a,b}]$, is unimodal.
Moreover, $S_{a,b} - \mathbb{E}[S_{a,b}]$ is symmetric about $0$, so the mode is at $0$ when $\Lambda = \mathbb{Z}$, and the modes are at $\pm  1/2$ when $\Lambda = \mathbb{Z} + 1/2$.
\end{proof}

Unimodality of $S_{a,b}$ implies $p(a,b; \Delta)$ is monotone in $a,b$ along feasible directions.

\begin{proof}[Proof of the second case of Lemma \ref{lem:ab_monotone}]
We prove the second case $(a',b') = (a, b+2)$.
In this case, the random variable $S_{a,b'}$ has the same distribution as $S_{a,b} + \zeta'_1 + \zeta'_2$,
where $\zeta'_1$ and $\zeta'_2$ are independent half-Rademacher random variables and are independent of $S_{a,b}$. 
If $\Delta + 1 \in \Lambda_{a,b}$, then
\begin{align*}
p(a,b') = 2 Q_{a,b'}(\Delta)
= 2\left( \tfrac{1}{4} Q_{a,b}(\Delta-1) + \tfrac{1}{2} Q_{a,b}(\Delta) + \tfrac{1}{4} Q_{a,b}(\Delta+1) \right) 
\ge 2 Q_{a,b}(\Delta) 
= p(a,b),
\end{align*}
where the inequality follows from Lemma \ref{lem:unimodal}.
If $\Delta + 1 \notin \Lambda_{a,b}$, then
\begin{align*}
p(a,b')
= \tfrac{1}{2} Q_{a,b}(\Delta-1) +  Q_{a,b}(\Delta)
= \tfrac{1}{2} \pi_{a,b, \Delta - 1} + \tfrac{3}{2} \pi_{a,b, \Delta}
\ge 2 \pi_{a,b, \Delta} = p(a,b).
\end{align*}    
\end{proof}

\begin{proof}[Proof of Lemma \ref{lem:g_amax}]
Let $\vv$ be a feasible outcome with $a(\vv) = a$ and $b(\vv) = b^{\max}(a)$. Consider another feasible outcome $\vv'$ with $a(\vv') = a+1$. Then, 
either $b^{\max}(a+1) = b^{\max}(a)$ or $b^{\max}(a+1) = b^{\max}(a) - 2 \ge 1$.
Let $(a',b')$ be the parameter pair for $\vv'$.

In the first case, by Lemma \ref{lem:ab_monotone}, $p(a',b'; \Delta) \ge p(a,b;\Delta)$ and thus $g(a+1) \ge g(a)$.

In the second case, write 
\[
S_{a,b} = R + H, \quad \text{where } H = \zeta_1 + \zeta_2,
\]
and let $\xi$ be a Rademacher random variable that is independent of $S_{a,b}$.
Then, $S_{a',b'}$ has the same distribution as $R + \xi$.

Suppose $\Delta + 1$ is in the support of $R$. 
Note that 
\begin{align*}
\Pr(S_{a,b} \ge \Delta) & = \Pr(R + \zeta_1 + \zeta_2 \ge \Delta) \\
& = \frac{1}{4} \Pr(R \ge \Delta - 1 ) + \frac{1}{2} \Pr(R \ge \Delta) + \frac{1}{4} \Pr(R \ge \Delta+1),
\end{align*}
and 
\begin{align*}
\Pr(S_{a',b'} \ge \Delta) & = \Pr(R + \xi \ge \Delta) \\
& = \frac{1}{2} \Pr(R \ge \Delta - 1) + \frac{1}{2} \Pr(R \ge \Delta + 1).
\end{align*}
Taking the difference, 
\begin{align*}
\Pr(S_{a',b'} \ge \Delta) -  \Pr(S_{a,b} \ge \Delta)
& = \frac{1}{4} \Pr(R \ge \Delta-1) - \frac{1}{2} \Pr(R \ge \Delta) + \frac{1}{4} \Pr(R \ge \Delta+1)
\ge 0.
\tag{by Lemma \ref{lem:unimodal}}
\end{align*}
Thus, in this case, $g(x)$ is nondecreasing.

Now, suppose $\Delta+1$ is not in the support of $R$, that is, $\Delta$ is the maximum number in the support of $R$. Thus,
\begin{align*}
\Pr(S_{a',b'} \ge \Delta) -  \Pr(S_{a,b} \ge \Delta) & = \frac{1}{4} \Pr(R \ge \Delta-1) - \frac{1}{2} \Pr(R \ge \Delta)    \\
& = \frac{1}{4} \Pr(R = \Delta - 1) - \frac{1}{4} \Pr(R = \Delta) \\
& \ge 0.
\tag{by Lemma \ref{lem:unimodal}}
\end{align*}
Thus, in this case, $g(x)$ is nondecreasing.
\end{proof}

Lemma \ref{lem:monotone_pab_b0} considers the special case $b = 0$.

\begin{proof}[Proof of Lemma \ref{lem:monotone_pab_b0}]
Since $b = 0, \tau_0 = 0$, $a$ and $\Delta = \frac{n}{2} T^{\obs}$ have the same parity.
If $\Delta = 1$, then $p(a,0;\Delta) = \Pr(S_{a,0} \ge \Delta) = 1$ for any feasible $a$.
In the rest of the proof, we assume $\Delta \ge 2$.

Consider two feasible outcome count vectors $\vv$ and $\vv'$ satisfying $b(\vv) = b(\vv') = 0$ and $a(\vv') = a(\vv) + 2$. If $\Delta + 2$ is in the support of $S_{a,0}$, then
\begin{align*}
p(a',0; \Delta) = 2Q_{a',0}(\Delta) 
= \frac{1}{2} Q_{a,0} (\Delta - 2) + Q_{\vv}(\Delta) + \frac{1}{2} Q_{a,0} (\Delta + 2)
\ge 2 Q_{a,0} (\Delta)
= p(a,0; \Delta).
\end{align*}
If $\Delta + 2$ is not in the support of $S_{a,0}$, then 
\[
p(a',0;\Delta) = \frac{1}{2} Q_{a,0} (\Delta - 2) + Q_{a,0}(\Delta)
= \frac{1}{2} \pi_{a,0, \Delta - 2} + \frac{3}{2} \pi_{a,0, \Delta} \ge 2\pi_{a,b, \Delta} = p(a,0;\Delta).
\]
\end{proof}

\subsection{Missing Proofs and Details in Section \ref{sec:ab_max}}

We show how to derive Equation \eqref{eqn:constraints_in_ab}, the constraints on $(a,b)$.

\paragraph{Deriving the first equation.}
By the first four type-(1) constraints on $\vv$ in Equation \eqref{eqn:constraints_in_v_observe_data}, we obtain
\begin{align*}
v_{11}+v_{10}\ge N_{11}
&\iff a+\frac{b+n\tau_0}{2}\ge N_{11}
\iff 2a+b\ge 2N_{11}-n\tau_0,\\
v_{11}+v_{01}\ge N_{10}
&\iff a+\frac{b-n\tau_0}{2}\ge N_{10}
\iff 2a+b\ge 2N_{10}+n\tau_0, \\
v_{01}+v_{00}\ge N_{01}
&\iff \frac{b-n\tau_0}{2}+(n-a-b)\ge N_{01}
\iff 2a+b\le 2n-n\tau_0-2N_{01},\\
v_{10}+v_{00}\ge N_{00}
&\iff \frac{b+n\tau_0}{2}+(n-a-b)\ge N_{00}
\iff 2a+b\le 2n+n\tau_0-2N_{00}.
\end{align*}
Hence
\[
\max\{2N_{11}-n\tau_0,\;2N_{10}+n\tau_0\}  \le 2a+b \le
\min\{2n-n\tau_0-2N_{01},\;2n+n\tau_0-2N_{00}\}.
\]

\paragraph{Deriving the rest equations.}
From the nonnegativity and the last four type-(1) constraints on $\vv$ in Equation \eqref{eqn:constraints_in_v_observe_data}, we have 
\[
0\le v_{11}=a \le N_{11}+N_{10}.
\]
In addition, we have 
\[
v_{00}=n-a-b\ge 0 \iff a+b\le n,
\]
and
\[
v_{00}=n-a-b\le N_{01}+N_{00}\iff a+b\ge n-(N_{01}+N_{00})=N_{11}+N_{10}.
\]
So
\[
N_{11}+N_{10}\le a+b\le n.
\]

By lower bound and parity for $b$, we have
\[
b\ge |n\tau_0|,\qquad b\equiv n\tau_0\pmod 2.
\]

Using the upper bounds on $v_{10}$ and $v_{01}$, we have the following upper bounds in $b$:
\begin{align*}
v_{10}\le N_{11}+N_{00}
&\iff \frac{b+n\tau_0}{2}\le N_{11}+N_{00}
\iff b\le 2N_{11}+2N_{00}-n\tau_0,\\
v_{01}\le N_{01}+N_{10}
&\iff \frac{b-n\tau_0}{2}\le N_{01}+N_{10}
\iff b\le 2N_{01}+2N_{10}+n\tau_0.
\end{align*}
Therefore
\[
b \le \min\{ 2N_{11}+2N_{00}-n\tau_0,\;2N_{01}+2N_{10}+n\tau_0 \}.
\]

\paragraph{Equivalence (both directions).}
We have shown the constraints in $\vv$ implies the constraints in $(a,b)$ in Equation \eqref{eqn:constraints_in_ab}
by substitution. Conversely, if $(a,b)$ satisfies Equation \eqref{eqn:constraints_in_ab}, then all inequalities above follow by reversing each algebraic step, and integrality/nonnegativity follow from $b\equiv n\tau_0 \pmod 2$ and $b\ge |n\tau_0|$.
Hence, the constraints are equivalent.

Finally, we prove the runtime of the FFT algorithm for computing $p(a,b; \Delta)$ given $a,b,\Delta$.

\begin{proof}[Proof of Lemma \ref{lem:Pa_fft}]
We use the formula of $p(a,b;\Delta) = \Pr(|S_{a,b}| \ge \Delta)$ in Lemma \ref{lem:objective_in_v}.
Let
\[
a = v_{11}, \quad b = v_{10} + v_{01, \quad}
M = 2a+b.
\]
Let $S = S_{a,b}$.
The support of $S$ is $\{-\frac{M}{2}, -\frac{M}{2}+1, \ldots, \frac{M}{2} \}$.
Define
\[
q_r \defeq \Pr\!\left(S=\frac{r}{2}\right), \qquad r\in\{-M,-M+1,\dots,M\}.
\]
We will compute all $q_r$ using FFT.

For $\theta\in\mathbb{R}$ define the characteristic function (where $i= \sqrt{-1}$)
\begin{align}
\psi(\theta) \defeq \mathbb{E}\!\left[e^{i\theta S}\right] = (\cos\theta)^{a}\,(\cos(\theta/2))^{b}.
\label{eq:psi_closed_form}    
\end{align}
Choose $N$ to be the smallest power of two with $N\ge 2M+1$.
Define
\[
\theta_j \defeq \frac{4\pi j}{N},
\qquad
F_j \defeq \psi(\theta_j), \qquad
j=0,1,\dots,N-1
\]
Using \eqref{eq:psi_closed_form}, all $F_j$ can be evaluated in $O(N \log n)$ arithmetic operations (where $\log n$ for the power).

Because $S$ takes values $r/2$ with probabilities $q_r$,
\[
F_j
= \sum_{r=-M}^{M} q_r\, e^{i\theta_j (r/2)}
= \sum_{r=-M}^{M} q_r\, e^{i (4\pi j/N)(r/2)}
= \sum_{r=-M}^{M} q_r\, e^{i 2\pi j r/N}.
\]
Now define a length-$N$ sequence $a=(a_0,\dots,a_{N-1})$ by
\[
a_k \defeq
\begin{cases}
q_{k-M}, & 0\le k\le 2M,\\
0, & 2M<k\le N-1.
\end{cases}
\]
Then, for each $j$,
\[
\sum_{k=0}^{N-1} a_k\, e^{i2\pi j k/N}
= \sum_{k=0}^{2M} q_{k-M}\, e^{i2\pi j k/N}
= e^{i2\pi j M/N}\sum_{r=-M}^{M} q_r\, e^{i2\pi j r/N}
= e^{i2\pi j M/N} F_j \defeq A_j.
\]
Thus, $\{A_j\}_{j=0}^{N-1}$ is exactly the length-$N$ discrete Fourier transform (DFT) of $a$.
By the inverse DFT formula,
\[
a_k = \frac{1}{N}\sum_{j=0}^{N-1} A_j\, e^{-i2\pi j k/N},
\qquad k=0,1,\dots,N-1.
\]
Hence all $a_k$ can be computed from $\{A_j\}$ using one inverse FFT in $O(N\log N)$ operations (see Chapter 30 of \citet{thomas2009introduction}),
and then
\[
q_r = a_{r+M}, \qquad r=-M,-M+1,\dots,M.
\]

Finally,
\[
p(\vv) = \Pr(|S|\ge |\Delta|)
= \sum_{r=-M}^{M} \mathbbm{1}\!\left\{\left|\frac{r}{2}\right|\ge |\Delta|\right\} q_r,
\]
which is an $O(M)=O(N)$ summation once $\{q_r\}$ is known.

We have $M=2a+b \le 2n$, so the chosen power of two satisfies $N=\Theta(M)=O(n)$.
The total cost is
\[
O(N \log n) \ \text{(forming $F_j$ and $A_j$)} \;+\; O(N\log N)\ \text{(one inverse FFT)} \;+\; O(N)\ \text{(tail sum)}
= O(n\log n).
\]
This proves the lemma.
\end{proof}

\section{Fast Exact Confidence Set Construction for General Bernoulli Design}
\label{sec:general_bernoulli}

In this section, we extend our fast confidence interval construction for the Bernoulli design from the balanced setting to the more general setting in which each unit is independently assigned to treatment with probability $p \in (0,1)$ and to control with probability $1-p$, that is, $\Pr(Z_i = 1) = p$ and $\Pr(Z_i = 0) = 1-p$ for all $i$.

When $p \neq 1/2$, the monotonicity of $p_{\max}(\tau_0)$ established in Lemma~\ref{lem:f_monotone} no longer holds, so a binary search over $\tau_0 \in \mathcal{T}$ is no longer applicable. Moreover, for a fixed $\tau_0 \in \mathcal{T}$, the monotonicity of the $p$-value with respect to $(a,b)$ in Lemma~\ref{lem:ab_monotone} also fails, which necessitates enumerating all feasible $(a,b)$ pairs to compute $p_{\max}(\tau_0)$.

\citet{aronow2023fast} encounter the same difficulties when extending their fast confidence interval construction for complete randomization from the balanced setting to the more general setting. To speed up the brute-force search over the space of feasible parameters, they propose reusing their Monte Carlo permutation test results to compute $p$-values for ``nearby'' potential outcome configurations, rather than recalculating them from scratch.

We build on the above idea to speed up our enumeration.
However, we will use the FFT to compute the exact probabilities, instead of the Monte Carlo method.

For a general assignment probability $p$, under a given outcome count vector $\vv$, let $\Tilde{\YY}$ be a outcome table consistent with $\vv$ and $\Tilde{\ZZ}$ be a random assignment drawn from the matched-pairs design. We can write the centered test statistic as
\[
n \, (T(\Tilde{\YY}, \Tilde{\ZZ}) - \tau_0) = 
\sum_{r=1}^{v_{11}} \eta_r + \sum_{s=1}^{v_{10}} \eps_s + \sum_{t=1}^{v_{01}} \delta_t \defeq S_{\vv},
\]
where $\eta_r \in \{\frac{1}{p}, - \frac{1}{1-p}\}$, $\eps_s \in \{\frac{1-p}{p}, -1 \}$, and $\delta_t \in \{1, -\frac{p}{1-p}\}$ are independent mean-zero random variables. 
Then, 
\[
p(\vv, \Delta) = \Pr(|S_{\vv}| \ge \Delta), \quad \Delta = n \, | T^{\obs} - \tau_0 |.
\]
We can evaluate $p(\vv, \Delta)$ and compute the probability mass function (PMF) of $S_{\vv}$ using FFT.

Now, we consider a feasible local $\vv'$:
\[
(v'_{11}, v'_{10}, v'_{01}, v'_{00}) = 
(v_{11}, v_{10} + 1, v_{01} - 1, v_{00}).
\]
The PMF of $S_{\vv'}$ can be computed as follows.
Note that 
\[
S_{\vv'} \overset{d}{=} S_{\vv} - \delta + \eps,
\]
where $\delta$ is an independent copy of $\delta_1$ and $\eps$ is an independent copy of $\eps_1$.
Then,
\begin{align}
\begin{split}
\Pr(S_{\vv'} = t) & = p^2 \, \Pr(S_{\vv} = t + 1 - \tfrac{1-p}{p})    
+ p(1-p) \, \Pr(S_{\vv} = t+2 \text{ or } -\tfrac{2p}{1-p}) \\
& \quad
+ (1-p)^2 \, \Pr(S_{\vv} = t - \tfrac{p}{1-p} + 1).
\end{split}
\label{eqn:pmf_update}
\end{align}
Thus, given the PMF of $S_{\vv}$, we can compute the PMF of $S_{\vv'}$ in time $O(n)$.

We describe the algorithm for general Bernoulli designs in Algorithm \ref{alg:ci_general}.
We enumerate all $\tau_0 \in \mathcal{T}$. 
Under a fixed $\tau_0$, we enumerate all feasible choices of $v_{11}$, compute the PMF for the smallest feasible $v_{10}$, and then compute the PMF for the rest choices of $v_{10}$ using Equation \eqref{eqn:pmf_update}.

The total number of $p$-values computed via FFT is $O(n^2)$. The total runtime is $O(n^3 \log n)$.
This completes the proof of Theorem \ref{thm:general}.

\begin{algorithm}[t]

	\SetAlgoNoLine
	\KwIn{Observed data $\NN$ and confidence level $\alpha \in (0,1)$.}
	\KwOut{A $(1-\alpha)$ confidence set $\mathcal{I}_{\alpha}$.}

Let $\mathcal{I}_{\alpha} \gets \emptyset$.

Let $\mathcal{T}$ be the feasible set of the ATE in Equation \eqref{eqn:tau_set}.
Let $n \leftarrow N_{11} + N_{10} + N_{01} + N_{00}$.
Compute the observed test statistic $T^{\obs}$.

\For{each $\tau_0 \in \mathcal{T}$}{
    Let $\Delta \gets n\abs{T^{\obs} - \tau_0}$.

Let $P \gets \emptyset$ 
\tcp*{the set of $p$-values under $\tau_0$}

\For{each feasible $v_{11}$ under $(\tau_0, \NN)$}{
    Let $v_{10}$ be the smallest feasible number under $(\tau_0, \NN, v_{11})$. 
    
    Let $v_{01} \gets v_{10} - n\tau_0$ and $v_{00} \gets n - v_{11} - v_{10} - v_{01}$.

    Compute $p(\vv, \Delta)$ and the PMF of $S_{\vv}$ via FFT.

    Let $P \gets P \cup \{p(\vv, \Delta)\}$.

    Let $v_{10} \gets v_{10} + 1$ and update $v_{01}, v_{00}$ accordingly.

    \While{$\vv$ is feasible under $(\tau_0, \NN)$}{
    
    Compute the PMF of $S_{\vv}$ from the PMF of the previous $\vv$ using Equation \eqref{eqn:pmf_update}.
    Compute $p(\vv, \Delta)$ using the PMF.

    Let $P \gets P \cup \{p(\vv, \Delta)\}$.

    Let $v_{10} \gets v_{10} + 1$ and update $v_{01}, v_{00}$ accordingly.
    }
}

If $\max P \ge \alpha$, then $\mathcal{I}_\alpha \gets \mathcal{I}_\alpha \cup \{\tau_0\}$.

}

\Return $\mathcal{I}_{\alpha}$.

	\caption{Confidence Set for General Bernoulli}
	\label{alg:ci_general}  
\end{algorithm}

\end{document}